\documentclass[conference]{IEEEtran}
\usepackage{graphicx}
\usepackage{amsmath}
\usepackage{lipsum}
\usepackage{amsthm}
\usepackage[strings]{underscore}
\usepackage{amsfonts}
\usepackage{caption}
\usepackage{siunitx}
\usepackage{amssymb}
\usepackage{subcaption}
\usepackage{placeins}
\usepackage{lettrine}
\usepackage{fancyhdr} 
\usepackage{multicol, blindtext}
\usepackage{bbm}
\usepackage{calrsfs}
\usepackage{cite}
\usepackage{color}
\usepackage{cuted}
\usepackage{mathtools}
\DeclarePairedDelimiter\ceil{\lceil}{\rceil}

\usepackage{multicol}
\usepackage{enumitem}
\usepackage{hyperref}
\usepackage{algorithm}
\usepackage{algpseudocode}
\usepackage{authblk}
\makeatletter
\newcommand\blfootnote[1]{%
  \begingroup
  \renewcommand\thefootnote{}\footnote{#1}%
  \addtocounter{footnote}{-1}%
  \endgroup
}

\newcommand{\Rmnum}[1]{\expandafter\@slowromancap\romannumeral #1@}

\newtheorem{theorem}{Theorem}
\newtheorem{remark}{Remark}

\newtheorem{proposition}{Proposition}
\newtheorem{conclusion}{Conclusion}

\newtheorem{lemma}{Lemma}
\makeatletter
\def\BState{\State\hskip-\ALG@thistlm}
\makeatother
\setlist[itemize]{leftmargin=*}
\begin{document}
%
\title{The Age of Incorrect Information: an Enabler of Semantics-Empowered Communication}
\author[$\S$]{Ali Maatouk}
\author[*]{Mohamad Assaad}
\author[$\dagger$]{Anthony Ephremides}
\affil[$\S$]{Paris Research Center, Huawei Technologies, Boulogne-Billancourt, France}
\affil[*]{TCL Chair on 5G, Laboratoire des Signaux et Syst\`emes, CentraleSup\'elec, Gif-sur-Yvette, France }
\affil[$\dagger$]{ECE Dept., University of Maryland, College Park, MD 20742}
\maketitle
\thispagestyle{fancy}
\pagestyle{fancy}
\fancyhf{}
\fancyheadoffset{0cm}
\renewcommand{\headrulewidth}{0pt} 
\renewcommand{\footrulewidth}{0pt}
\renewcommand*{\thepage}{\scriptsize{\arabic{page}}}
\fancyhead[R]{\thepage}
\fancypagestyle{plain}{%
   \fancyhf{}%
   \fancyhead[R]{\thepage}%
}

\begin{abstract}
In this paper, we introduce the Age of Incorrect Information (\textbf{AoII}) as an enabler for semantics-empowered communication, a newly advocated communication paradigm centered around data's role and its usefulness to the communication's goal. First, we shed light on how the traditional communication paradigm, with its role-blind approach to data, is vulnerable to performance bottlenecks. Next, we highlight the shortcomings of several proposed performance measures destined to deal with the traditional communication paradigm's limitations, namely the Age of Information (\textbf{AoI}) and the error-based metrics. We also show how the AoII addresses these shortcomings and captures more meaningfully the purpose of data. Afterward, we consider the problem of minimizing the average AoII in a transmitter-receiver pair scenario. We prove that the optimal transmission strategy is a randomized threshold policy, and we propose an algorithm that finds the optimal parameters. Furthermore, we provide a theoretical comparison between the AoII framework and the standard error-based metrics counterpart. Interestingly, we show that the AoII-optimal policy is also error-optimal for the adopted information source model. Concurrently, the converse is not necessarily true. Finally, we implement our policy in various applications, and we showcase its performance advantages compared to both the error-optimal and the AoI-optimal policies. \blfootnote{A preliminary version of this work has been presented at the 2022 IEEE International Conference on Communications \cite{9838737}.}
\end{abstract}




%

\section{Introduction}
\color{black}
In the last decade, communication systems have witnessed astronomical growth in both traffic demand and widespread deployment. Thanks to the technological advances in battery productions and the cheap cost of radio-enabled devices, communication systems are no longer constrained to the traditional data and voice exchange frameworks. Today, wireless devices provide essential services and play a vital role in various disciplines. For example, the Internet of Things (\textbf{IoT}) revolution is reshaping modern healthcare systems by incorporating technological, economic, and social prospects. This was witnessed lately amid the global COVID-19 pandemic, where wireless devices for tracking and collecting patient data were prevalent. This example barely scratches the surface as IoT systems are gaining massive momentum in many other domains. 
Given that we are just witnessing the tip of the iceberg, a natural question arises: are current communication paradigms suitable to deal with such demand? Furthermore, are we extracting the best possible performance from the communication networks?

Like any system, these networks' performance is contingent on the performance measure's choice that we set our goal to optimize. Traditionally, metrics like throughput, delay, and packet loss were adopted. Note that these metrics do not consider the packets' content and the amount of information they bring to the destination. Therefore, we can see that traditional communication paradigms follow a blind approach to data packets' content at both the physical and data link layers. 
In other words, at these layers, packets are treated equally regardless of the amount of information they will potentially bring to the destination.
 Given the anticipated astronomical growth in traffic demand and the potential interconnections between these systems, this content-blind approach to network optimization can lead to performance bottlenecks. Accordingly, researchers have been trying to push the boundaries of this traditional paradigm and establish more elaborate frameworks for network optimization. Perhaps one of the most recent successful efforts was the introduction of the Age of Information (\textbf{AoI}) \cite{6195689}. The AoI quantifies the notion of information freshness by measuring the information time lag at the destination. By incorporating this metric in the network's optimization, we give another dimension to the data packets as they will no longer be treated equally at these layers. For example, a packet is given more importance when its destination has not been updated for a while. Following its introduction, a surge in the number of papers on the AoI can be seen (we refer the readers to \cite{9380899,maatouk:tel-03028195} for a literature review). This surge is due to the expected performance improvement this added dimension will have in various applications (e.g., \cite{8000687,9241401,9007478,8406945,8845254,2018arXiv180103975J,8514816}). 

Although the AoI was shown to provide significant improvements to data freshness in various applications, it exhibits some critical shortcomings. Precisely, the AoI infers the importance of packets through their timestamps only and does not consider their content. Due to this property, recent works showed that age-optimal sampling policies are incapable of minimizing the prediction/mean squared error in remote estimation applications \cite{8812616}. 
 Given this shortcoming of the AoI, researchers have proposed data acquisition and scheduling schemes based on error minimization and the notion of the value of information in control theory (e.g., \cite{8812616,8737404,10.1145/3302509.3311050}). The adoption of error-based metrics in data acquisition and transmission decisions at the PHY/MAC layers allows us to abolish the separation principle prevalent in the traditional communication frameworks. Even though this is a step forward in the right direction, error-based metrics come short in capturing a crucial aspect of the communication: its goal. In fact, these metrics do not consider what the packets are used for, but rather their optimization aims solely to reduce the mismatch between the physical process and its estimate at the destination. Given that the communication's goal is neglected, adopting these metrics could hinder achieving the desired goal.

To address these shortcomings, the present authors and several other researchers have been recently advocating for a new communication paradigm based on the notion of ``Semantics of Data" \cite{9137714,9083812,9475174}. The framework of semantics has been previously proposed in \cite{ShannonWeaver49} for the case of language communication. In this framework proposed in
1949, which is suitable for voice/text/images-related applications, the importance of a message consists of its contribution to the meaning that wants to be conveyed to the distant receiver. For other types of applications, such as real-time monitoring and Machine-to-Machine applications, semantics of data is evaluated with respect to goal oriented metrics that capture the receiver’s utility for information. In other words, semantics of data is employed here to express the data significance and usefulness to the communication's goal. \color{black}To understand this concept, let us consider an example of a communication network involving various temperature sensors and a central controller. In these settings, the goal is not to always have timely packets delivered about the sensors' temperature processes nor to minimize the mismatch between the temperature processes and their estimates at the controller.
On the contrary, the sole goal is to make sure the controller reacts swiftly to any abnormal temperature rise. Therefore, to extract the best performance out of the network, our system's design must undoubtedly include the purpose of the data involved. In this case, when sampling or transmitting packets, we look at the bigger picture of how vital these packets are to achieve our prescribed goal. Using the notion of data semantics, the objective is to establish a network optimization framework that is adaptable to any communication goal by merely changing a set of parameters of a general performance metric. This brings us to the new notion of Age of Incorrect Information (\textbf{AoII}), proposed by the present authors in \cite{9137714} that can be considered as a step toward that ultimate goal.

The AoII was introduced to address the shortcomings of both the AoI and the error-based metrics by incorporating the semantics of data more meaningfully. Specifically, the AoII is a proposed performance measure that captures the significance of a packet within a specific general communication goal through two aspects 1) an information-penalty aspect and 2) a time-aspect function. As will be seen in the remainder of the paper, by definition,  the AoII considers the content of packets, the information knowledge at the destination, and the effect of the mismatch between the physical process and its estimate on the overall communication's goal. Interestingly, we will show that many real-life applications' communication goals are merely variants of the AoII obtained by tweaking specific parameters. To that end, we summarize in the following the key contributions of this paper:
\color{black}
\begin{figure*}[ht]
\centering
\begin{subfigure}{0.33\textwidth}
  \centering
  \includegraphics[width=.9\linewidth]{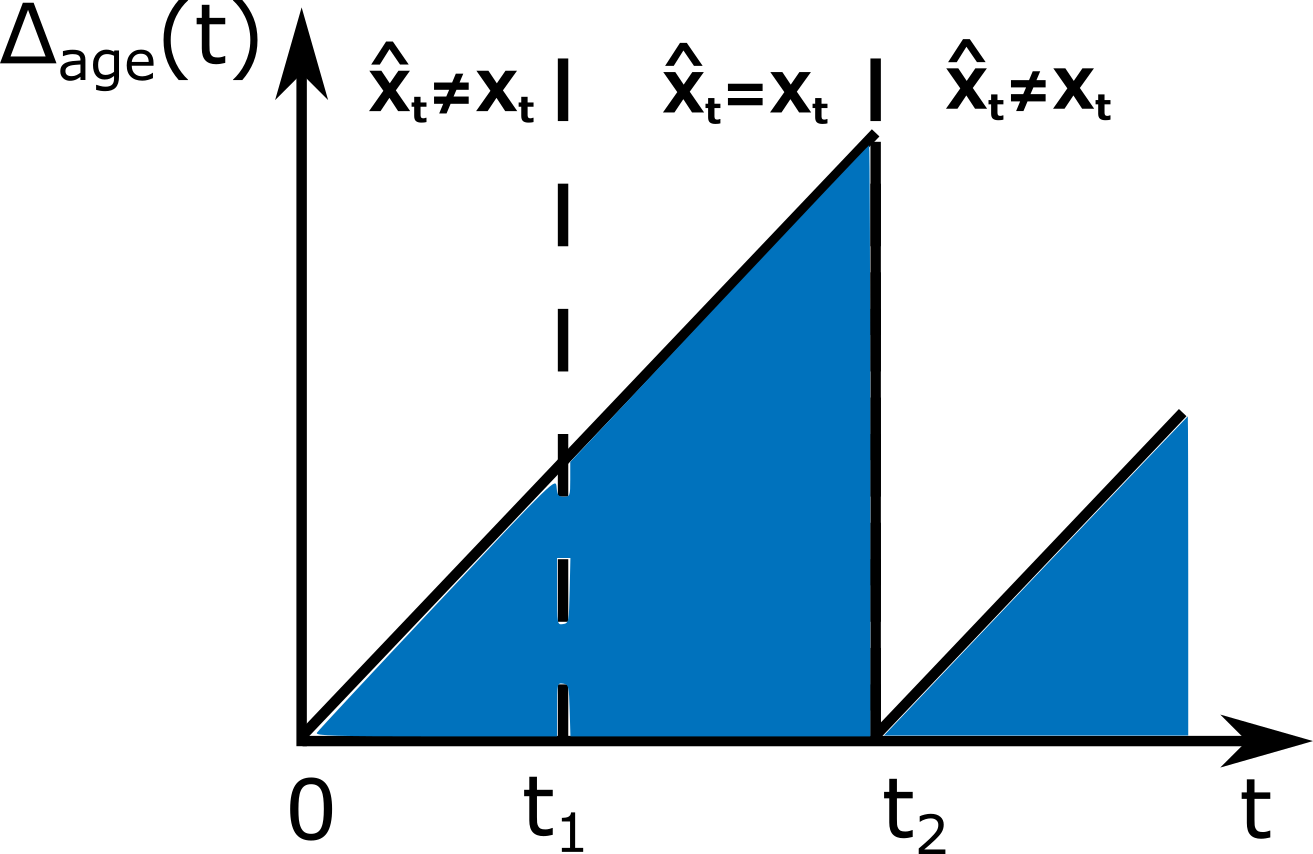}
  \caption{Age penalty function.}
    \label{agemetric}
\end{subfigure}%
\begin{subfigure}{0.33\textwidth}
\centering
  \includegraphics[width=.99\linewidth]{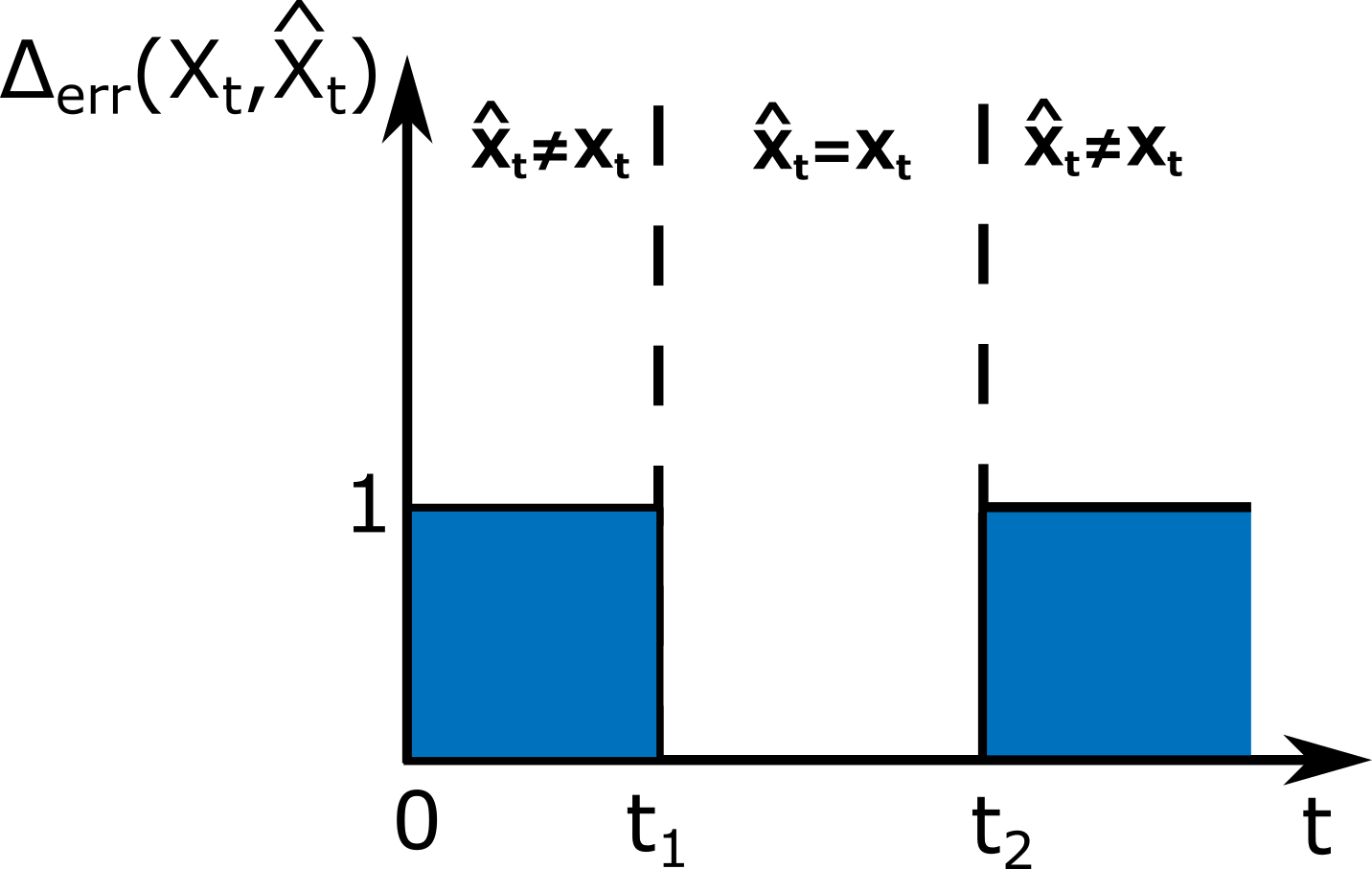}
  \caption{Error penalty function.}
\label{errmetric}
\end{subfigure}%
\begin{subfigure}{0.33\textwidth}
\centering
  \includegraphics[width=.99\linewidth]{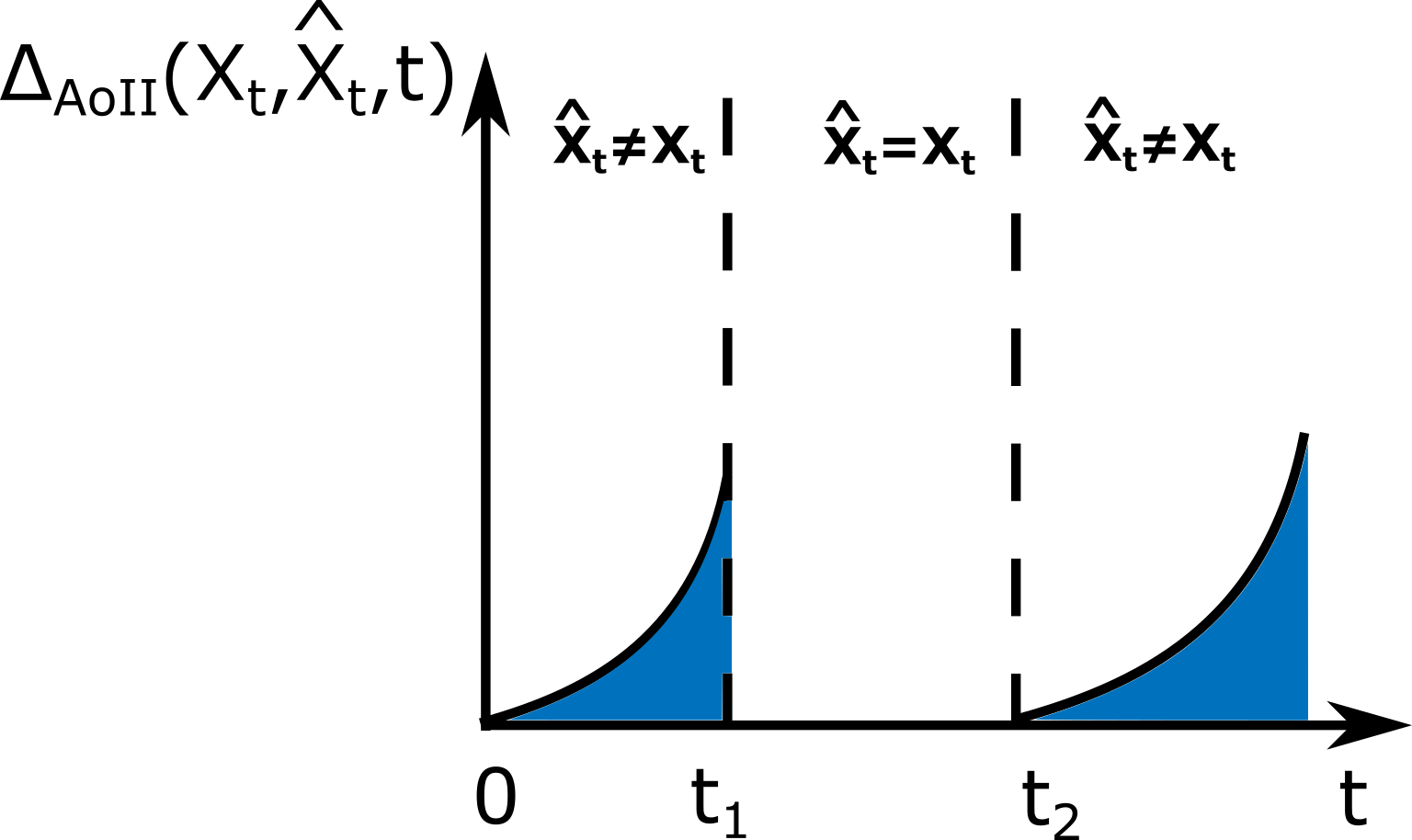}
  \caption{AoII penalty function.}
\label{effmetric}
\end{subfigure}%
\caption{Illustrations of the different penalty functions.}
\label{metrics}
\end{figure*}
\begin{itemize}
\item We consider the problem of minimizing the average AoII in a transmitter-receiver pair scenario where packets are sent over an unreliable channel subject to a transmission rate constraint. Compared to our previous work on the AoII \cite{9137714} where a linear version of the AoII was studied, we consider a more general version of the AoII where any non-decreasing dissatisfaction function $f(\cdot)$ can be adopted. This generalization leads to numerous technical challenges that we address in this paper. Particularly, in this paper, we adopt a different approach where 1) We provide structural results on the problem at hand for both unbounded and asymptotically bounded functions $f(\cdot)$, 2) In both cases, we derive an expression of the value function and the update rate for any threshold policy, 3) We leverage fundamental properties of the AoII to show that an optimal policy can be constructed through randomization, 4) Finally, we provide pseudocode of the optimal transmission policies and prove their logarithmic complexity.
\item Afterward, we provide a thorough comparison between the AoII framework and the standard error-based metrics counterpart. Since the seminal work of Sun et al. \cite{8812616}, a large part of the work on the AoI aimed to find connections between the AoI minimization framework and the standard MMSE (minimum mean squared error) and prediction error minimization frameworks (e.g., \cite{8845106}). Our work on the AoII provides a framework where we go beyond the AoI and the standard error metrics and focus directly on the communication goal, thus enabling semantics-empowered communications. One key question to answer is how such a framework compares with the traditional error frameworks. Kam et al., in one of their recent works \cite{9162726}, showcased numerically that the minimization of the AoII led to a minimization of the prediction error, hence increasing the importance to answer such a question. One key consequence of the generalization done in this paper was that we could answer that question within a theoretical framework. Curiously, our comparison leads to an interesting conclusion: for the adopted information source model, the AoII-optimal policy is also error-optimal. At the same time, the converse is not necessarily true. \color{black}
\item Lastly, we provide several real-life applications where the communication's goal can be formulated as an AoII minimization problem by adequately choosing $f(\cdot)$. Such applications allow us to frame the AoII as an enabler of semantics-empowered communication, which is a radical new communication paradigm that has been receiving significant attention recently for 6G networks (e.g., \cite{CALVANESESTRINATI2021107930}). For the applications mentioned above, we show how our approach achieves a significant performance advantage compared to the AoI and the standard error metrics frameworks. 
\end{itemize}
The rest of the paper is organized as follows: Section \ref{whyage}
is dedicated to the motivation behind the AoII. The system model, along with the dynamics of the AoII, are presented in Section \ref{modellll}. Section \ref{optimizationgeneralklshi} presents our optimization approach to the problem at hand, along with the main results of the paper. In Section \ref{comparisonwitherror}, we theoretically compare the AoII-optimal transmission policy to the error framework and provide a key comparison between them. In Section \ref{numericalresults}, we provide real-life applications that fall within our framework and showcase the advantages of the AoII compared to both the AoI and error-based approaches. Lastly, we conclude our paper in Section \ref{conclusionsss}.
\section{Why The Age of Incorrect Information?}
\label{whyage}
To understand the notion of AoII, it is best to consider a basic transmitter-receiver system where a process $X_t$ is observed by the transmitter. For example, $X_t$ can be a machine's temperature, a vehicle's velocity, or merely the state of a wireless channel. To that end, $X_t$ is subject to possible changes at any time instant $t$, and these changes have to be reported to the monitor (receiver) through the transmission of status updates packets. Using these packets, the monitor creates an estimate of $X_t$ at each time $t$, denoted by $\hat{X}_t$. The monitor uses these estimates to complete tasks, make decisions, or carry out commands. Therefore, it is easy to see that the system's performance is contingent on a proper estimation of $X_t$ at each time $t$. Ideally, we would like to have a perfect estimation where $\hat{X}_t=X_t$ at any time instant $t$. However, given many limiting factors, such as the delay in wireless channels, this is not feasible in practice. Accordingly, one must adopt a particular penalty/utility function for which its minimization/maximization helps us achieve the system's best possible performance.

Traditionally, wireless networks have been looked at as a content-agnostic data pipe. In other words, the content of the data packets and the role they play in the broader scope of an application at the receiver have been overlooked from a network optimization perspective. To that end, the conventional goal in the communication paradigms has been to merely optimize network-based metrics such as throughput or delay through a smart allocation of the available resources. However, this approach strips away the context from the data. Therefore, packets are treated as equally important, regardless of the amount of information they bring to the monitor. Given the astronomical growth in data demand, the ubiquitous wireless connectivity, and the abundance of remote monitoring applications, a more effective approach to network optimization has to be adopted. Accordingly, the research community has been intensively trying to propose new network optimization frameworks to achieve this efficacy. To this date, the proposed frameworks generally fall into one of the two following groups:
\begin{enumerate}
\item Age-based metrics framework
\item Error-based metrics framework
\end{enumerate}
First, let us discuss the age-based metrics framework. The AoI, or simply the age, is defined as \cite{6195689}
\begin{equation}
\Delta_{\text{age}}(t)=t-U_t,
\label{agedef}
\end{equation}
where $U_t$ is the timestamp of the last successfully received packet by the monitor at time $t$. Essentially, the AoI captures the information time-lag at the monitor. To that end, the minimization of age-based metrics like the time-average age has been widely regarded as a means to achieve freshness in communication \cite{6195689}. This approach's idea is that with a guarantee of fresh data at the monitor, one would expect an overall better system performance. As one can see, contrary to the throughput and delay frameworks, adopting the AoI as a network performance metric avoids the equal treatment of packets. In fact, in this framework, data packets have the highest value when they are fresh. Consequently, the AoI lets us infer the importance of a packet using its generation time. Although the AoI is a step forward in the right direction, we can witness its fundamental flaw in many applications. To put this flaw into perspective, let us consider a time interval $[t_1,t_2]$ in which $X_t=\hat{X}_t$. In other words, during this interval, the monitor has a perfect estimate of the information process $X_t$. As seen from the age definition (\ref{agedef}) and Fig. \ref{agemetric}, the system is still penalized even in this time-interval. Due to this unnecessary penalization of the system, we can expect a waste of vital resources on useless status updates. This flaw is inherent in the AoI definition as it does not consider the current value of the information process and its estimate at the monitor. For this reason, age-optimal sampling policies were found to be sub-optimal in many remote estimation applications (e.g., \cite{8812616}). This leads us to the next class of proposed optimization frameworks: the error-based metrics framework.
\begin{remark}
It is worth mentioning that several \textbf{time-based} metrics have been proposed in the literature to address various shortcomings of the AoI. For example, the Age of Synchronization (\textbf{AoS}), which measures the time-elapsed since a new update was generated, was introduced for caching systems \cite{8437927}. Although they address several shortcomings of the AoI, these metrics remain time-based and do not depend on the mismatch between $X_t$ and $\hat{X}_t$, which limits their usage in remote estimation applications. 
\end{remark}

The error-based metrics framework consists of taking as a network performance measure a quantitative representation of the difference between $\hat{X}_t$ and $X_t$. The hope is, by incorporating the information on $X_t$ and $\hat{X}_t$ in the performance metric, we can better utilize the available resources to let $\hat{X}_t$ be close to $X_t$. Among the most common error-based metrics, we have  
\begin{equation}
\Delta_{\text{err}}(X_t,\hat{X}_t)=\mathbbm{1}{\{\hat{X}_t\neq X_t\}},
\label{errorfunc}
\end{equation}
\begin{equation}
\Delta_{\text{sq}}(X_t,\hat{X}_t)=(X_t-\hat{X}_t)^2,
\label{functionalmmse}
\end{equation} 
where $\mathbbm{1}\{\cdot\}$ is the indicator function. By minimizing the time-average of the metrics found in (\ref{errorfunc}) and (\ref{functionalmmse}), we obtain the celebrated Minimum Prediction Error (\textbf{MPE}) and the Minimum Mean Squared Error (\textbf{MMSE}) policies respectively \cite{8812616,1582293,9174333}. It is clear that this framework does not have the AoI's fundamental shortcomings. For example, as illustrated in Fig. \ref{errmetric}, the penalty of the system is equal to $0$ in the time-interval $[t_1,t_2]$ in which $X_t=\hat{X}_t$. Additionally, one can notice that, similarly to the AoI framework, adopting an error-based metric as a network performance measure avoids the equal treatment of data packets. Interestingly, in this framework, data packets have the highest value when the difference between the information they carry and $\hat{X}_t$ is large. Although the error-based metrics add a sense of meaning to the packets compared to throughput and delay, they also have underlying flaws. As seen in (\ref{errorfunc})-(\ref{functionalmmse}), the error-based metrics only consider the difference between $X_t$ and $\hat{X}_t$ to infer the importance of the packets. Given that a perfect match $X_t=\hat{X}_t$ for all $t$ is not feasible in realistic scenarios, we can see that this approach fails to capture the effect their mismatch has on the overall communication's goal. To see this more clearly, let us consider that the information process $X_t\in\{0,1\}$ tracks the temperature of a machine. Let us suppose that $X_t=0$ indicates that the machine is operating at a normal temperature at time $t$ while $X_t=1$ indicates that it is overheating. We consider that the estimate $\hat{X}_t$ is used by the monitor to react to any sudden temperature spike in the machine. Now, let us assume that a spike occurs in the time interval $[0,t_1]$. As illustrated in Fig. \ref{errmetric}, the error-based metrics will lead to a constant penalization of the system. However, it is well-known from the physical characteristics of materials that an abnormal temperature rise's repercussions become more severe the longer that spike is prolonged. In the same spirit, this flaw is highlighted when we consider the phenomena of error bursts. As seen in Fig. \ref{bursterrors}, the system's error penalty due to two bursts of errors of one timeslot is equivalent to that resulting from a single error of two timeslots. However, it is well-known that in a large variety of applications, the repercussions of a long burst of error are far more severe (e.g., video streaming \cite{4490280}). 
\begin{figure}[!ht]
\centering
\includegraphics[width=.99\linewidth]{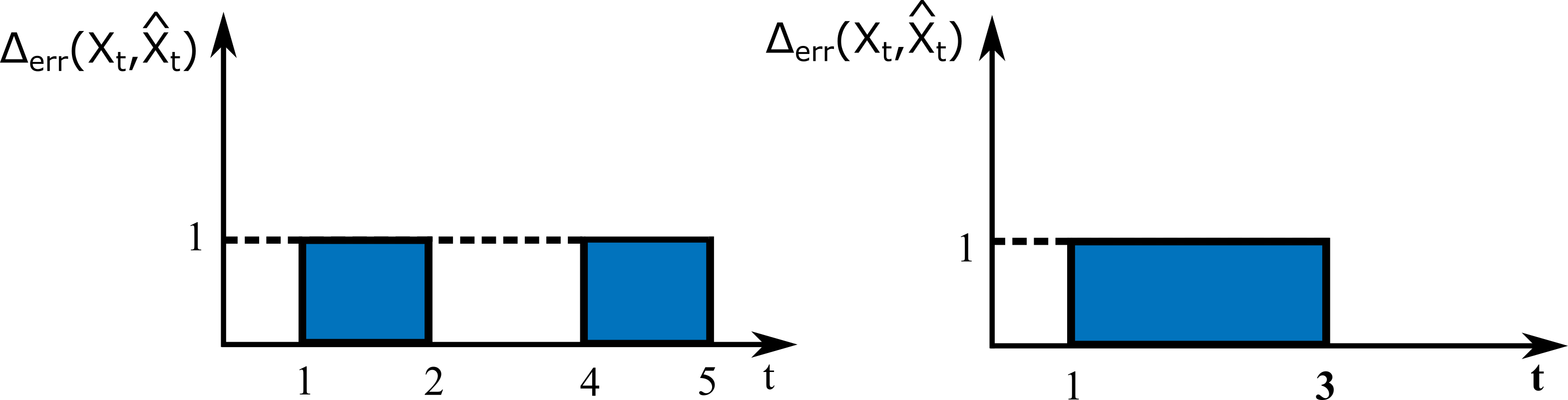}
\setlength{\belowcaptionskip}{-12pt}
\caption{Illustration of the burst errors situation.}
\label{bursterrors}
\end{figure}

Therefore, a better performance measure takes into account, not just the mismatch between $X_t$ and $\hat{X}_t$, but also how long that mismatch has been prevailing. By adopting such a metric, we capture more the context of data and their purpose. Accordingly, we can then enable semantics-empowered communication in the network, which is more elaborate than the AoI and the error-based frameworks. This leads us to our proposed metric: the AoII. We define the AoII as
\begin{equation}
\Delta_{\text{AoII}}(X_t,\hat{X}_t,t)=f(t)\times g(X_t,\hat{X}_t),
\label{propfunc}
\end{equation}
where $f: [0,+\infty)\mapsto [0,+\infty)$ is a non-decreasing function and $g(X_t,\hat{X}_t): \mathcal{D}\times \mathcal{D}\mapsto [0,+\infty)$ where $\mathcal{D}$ is the state space of $X_t$. The AoII is therefore a combination of two elements:
\begin{enumerate}
\item A function $g(\cdot,\cdot)$ that reflects the gap between $X_t$ and $\hat{X}_t$.
\item A function $f(\cdot)$ that plays the role of increasingly penalizing the system the more prolonged a mismatch between $X_t$ and $\hat{X}_t$ is.
\end{enumerate}
To better understand the metric, let us go back to the machine temperature example. As seen in Fig. \ref{effmetric}, the AoII is $0$ in the time-interval $[t_1,t_2]$ in which no mismatch exists. In the interval $[0,t_1]$, we can see that, unlike the error-based metrics, we are penalizing the system more the longer the mismatch lasts. As we have previously explained, this allows us to capture the purpose of the data being transmitted more meaningfully. Given that the performance of a network designed to take into account the purpose of data will always outperform any semantic-blind network, we delve into more details in the proposed AoII metric. The proposed AoII metric is quite general and presents itself as an umbrella for a large variety of performance measures depending on the selected functions $f(\cdot)$ and $g(\cdot,\cdot)$. For example, we can adopt for the function $g(\cdot,\cdot)$ any of the standard error-based metrics such as
\begin{itemize}
\item The indicator error function:
\begin{equation}
g_{\text{ind}}(X_t,\hat{X}_t)=\mathbbm{1}\{X_t\neq\hat{X}_t\}.
\label{functional7}
\end{equation}
We can choose this function when any mismatch between $X_t$ and $\hat{X}_t$, regardless of how big it is, equally harms the system's performance.
\item The squared error function:
\begin{equation}
g_{\text{sq}}(X_t,\hat{X}_t)=(X_t-\hat{X}_t)^2.
\label{functional8}
\end{equation}
Choosing this function implies that the larger the gap between $X_t$ and $\hat{X}_t$ is, the more significant its impact on the system's performance is.
\item The threshold error function:
\begin{equation}
g_{\text{threshold}}(X_t,\hat{X}_t)=\mathbbm{1}\{|X_t-\hat{X}_t|\geq c\},
\label{functional9}
\end{equation}
where $c>0$ is a predefined threshold. This is an adequate choice when the system's performance is immune to small mismatches between $X_t$ and $\hat{X}_t$. 
\end{itemize}
Next, to provide examples of the function $f(\cdot)$, we first define $V_t$ as the last time instant where $g(X_t,\hat{X}_t)$ was equal to $0$. Specifically, $V_t$ is the last time instant where the monitor had \textbf{sufficiently} accurate information about the process $X_t$. With this notion in mind, we provide in the following a few examples of $f(\cdot)$.
\begin{itemize}
\item The linear time-dissatisfaction function:
\begin{equation}
f_{\text{linear}}(t)=t-V_t.
\label{functional1}
\end{equation}
This can be used when the performance degrades uniformly with time when a mismatch occurs. \color{black}
\item The time-threshold dissatisfaction function:
\begin{equation}
f_{\text{threshold}}(t)=\mathbbm{1}\{t-V_t\geq \zeta\},
\label{functional3}
\end{equation}
where $\zeta>0$. We can choose this function when the system's performance can tolerate the mismatch between $X_t$ and $\hat{X}_t$ for a certain duration $\zeta>0$. Afterward, a penalty is incurred.\color{black}
\end{itemize}
Depending on the application at hand, we can adopt an appropriate choice of $f(\cdot)$ and $g(\cdot,\cdot)$ to capture the data's purpose. 
In simple applications, one may be able to derive explicitly these functions $f(\cdot)$ and $g(\cdot,\cdot)$ that capture the time and information facets playing a role in data significance as will be seen in later sections. However, in more complicated scenarios, one would need to fit the functions $f(\cdot)$ and $g(\cdot,\cdot)$ using gathered or generated data on the application of interest. 
\section{Model and Formulation}
\label{modellll}
\subsection{System Model}
We consider in our work a transmitter-receiver system where time is assumed to be slotted and normalized to the slot duration (i.e., the slot duration is taken as $1$). The transmitter observes an information process, denoted by $\big(X_t\big)_{t\in\mathbb{N}}$, that can change values with time and the transmitter's goal is to send status updates to keep the receiver up-to-date on the process' values\color{black}. To understand how this system works, let us suppose that the transmitter decides to transmit a packet at time $t$. Therefore, a sample of $X_t$ is generated, and the transmission stage immediately begins. The packet is transmitted over an unreliable channel where transmission errors may occur. We suppose that the channel realizations are independent and identically distributed over the timeslots and follow a Bernoulli distribution. In particular, the channel realization $h_t$ is equal to $1$ if the packet is successfully decoded by the receiver side and is $0$ otherwise. Given the Bernoulli assumption, we define the transmission success and failure probabilities as $\Pr(h_t=1)=p_s$ and $\Pr(h_t=0)=p_f=1-p_s$ respectively.
This model is motivated by short-packets transmission in block fading wireless channels where a target rate $R_{\text{target}}$ is required (we refer the readers to \cite[Section~II]{9174163}). 
\color{black}If the transmission is successful, the status update is delivered at time $t+1$, and the transmitter receives an instantaneous Acknowledgement (\textbf{ACK}) packet. The quick delivery and reliability of the ACK packets is a widely used assumption as these packets are typically small\color{black}. Accordingly, their transmission time can be considered negligible \cite{8000687,8514816}. Note that if an ACK is not received at $t+1$, the transmitter discards the old packet and generates a new update if it opts for a new transmission. By leveraging this mechanism, the transmitter can have perfect knowledge of the packets that arrive at the receiver. It is worth noting that the violation of the negligible ACK latency or its reliability dictates their consideration in the design of any optimization framework. \color{black}

Using the information found in the received updates, the receiver can only construct an estimate of the information process, denoted by $\big(\hat{X}_t\big)_{t\in\mathbb{N}}$. Let us now focus on the model of $X_t$ that we will adopt in the sequel. Similar to \cite{8812616}, we consider that the receiver's estimate of the information source is
\begin{equation}
\hat{X}_t=X_{U_t},
\end{equation}
where $U_t$ is the timestamp of the last successfully received packet by
the receiver at time t. In other words, the receiver takes the last successfully received update as an estimate of the information source. Let us now define $\big(d_t\big)_{t\in\mathbb{N}}=\mathbbm{1}\{|X_t-\hat{X}_t|\geq c\}$, where $c>0$ is a predefined threshold. The threshold $c$ can be thought to be a barrier between two regimes: 1) a GOOD regime where the  mismatches between $X_t$ and $\hat{X}_t$ does not affect the performance, and 2) a BAD regime where the contrary takes place. Large values of $c$ suggest that the system can tolerate to a certain extent mismatches between $X_t$ and $\hat{X}_t$, while small values of $c$ suggest its sensitivity to these mismatches. An illustration of this distinction of regimes can be found in Fig. \ref{sourcechain1111}. 
\begin{figure}[!ht]
\centering
\includegraphics[width=.9\linewidth]{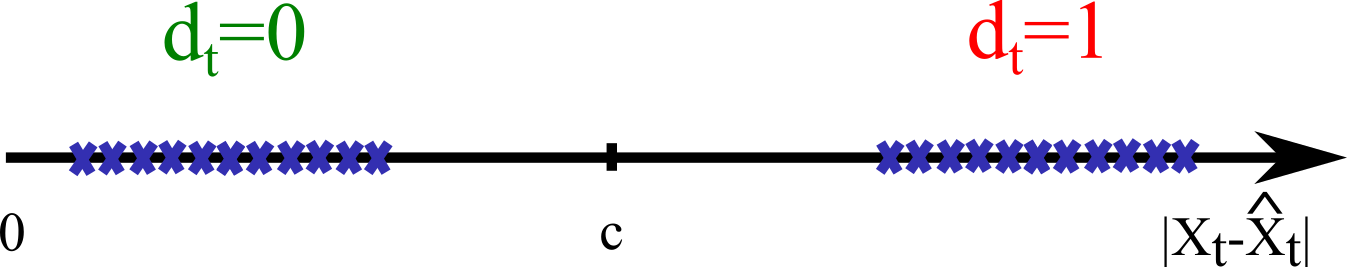}
\setlength{\belowcaptionskip}{-12pt}
\caption{Illustration of the regimes of the process $d_t$.}
\label{sourcechain1111}
\end{figure}\\
We focus in the sequel on the process $\big(d_t\big)_{t\in\mathbb{N}}$ due to its analytical advantage. In fact, tracking the exact continuous process' evolution, especially when adding a timing aspect as the AoII does, makes the theoretical analysis very challenging. This model accurately captures all the information about $X_t$ when $X_t$ typically experience either big changes or small changes but may come at a cost when we consider any general continuous process $X_t$. However, for tractability, we consider the former case in our paper.
\color{black}
Next, we tackle the modeling of $d_t$. Specifically, we consider that if no packets are delivered to the receiver, $\big(d_t\big)_{t\in\mathbb{N}}$ evolves as a $2$ states discrete Markov chain depicted in Fig. \ref{sourcechain} with parameters $\alpha$ and $\beta$. Note that these parameters will change based on the choice of $c$. \color{black}
\begin{figure}[!ht]
\centering
\includegraphics[width=.8\linewidth]{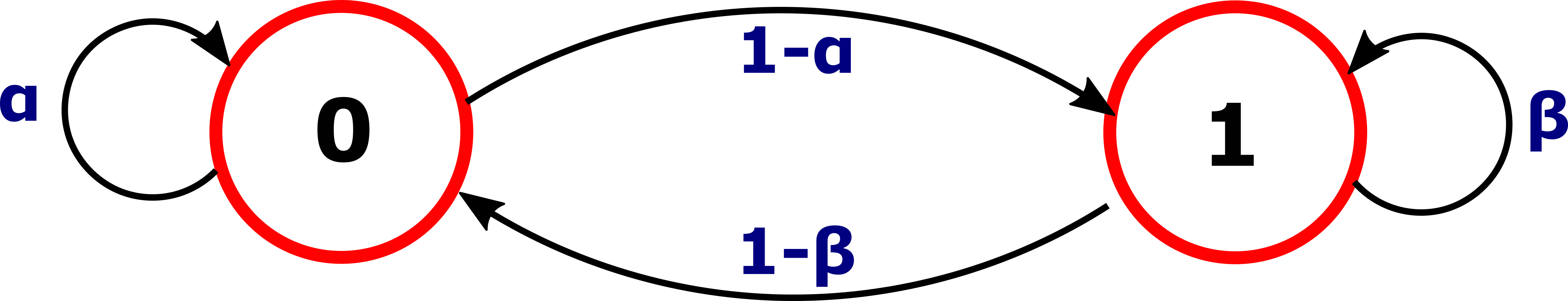}
\setlength{\belowcaptionskip}{-12pt}
\caption{Illustration of the process model.}
\label{sourcechain}
\end{figure}\vspace{3pt}\\
Although simple, this model encompasses a variety of real-life settings and have been adopted in numerous research works (e.g., \cite{8737404}). For instance, suppose that the observed process $X_t$ is a certain channel state and $\hat{X}_t$ is its estimate at the transmitter. By adopting a Markovian channel model, it can be shown that without any training using pilot symbols, $d_t$ can be modeled using a Markov chain similar to Fig. \ref{sourcechain}. Note that Markovian channels are a typical assumption for fading channels, and their usefulness is supported by experimental results (we refer the readers to \cite{350282}). On another note, Markov chains are also widely used to discretize and approximate continuous-valued processes (e.g., diffusion processes \cite{doi:10.1080/17442507908833139}, continuous-valued autoregression processes \cite{TAUCHEN1986177}). This puts in perspective the applicability of the adopted Markov chain model in various settings despite its simplicity. In addition, the simplicity of the model enables a better understanding of the dynamics and merits of the new performance measure.

Lastly, as is the case in realistic scenarios, we consider that the transmitter cannot send status updates at each timeslot. Precisely, due to battery limitations, for example, an average transmission frequency $\delta$ cannot be surpassed. Given the constraint on the transmission frequency and the random nature of the channel, the transmission policy's choice has an immense effect on the system's performance. As motivated in the previous subsection, we adopt the AoII as a performance measure of the system. To fully understand the evolution of the AoII, we provide details on its dynamics in the next subsection.
\subsection{System Dynamics}
\label{penaltydynamics}
In this paper, we focus on the class of AoII measures having the function $g(\cdot,\cdot)$ as $g_{\text{threshold}}(X_t,\hat{X}_t)=d_t$. To that end, let us define the system's state $S_t$ at time $t$ as
\begin{equation}
S_t=(t-V_t)d_t.
\end{equation}
Given that $t\in\mathbb{N}$, we have $S_t\in\mathbb{N}$. Next, as seen in the AoII examples given in (\ref{functional1})-(\ref{functional3}), the function $f(\cdot)$ is generally written in function of $t-V_t$. To that end, we consider in the sequel the class of non-decreasing dissatisfaction functions $f(\cdot)$ that can be written in terms of $t-V_t$. In this case, the AoII can be rewritten as $\Delta(X_t,\hat{X}_t,t)=d_tf(t-V_t)$. \color{black}With that in mind, we can rewrite the AoII as 
\begin{equation}
\Delta_{\text{AoII}}(X_t,\hat{X}_t,t)= f(S_t).
\end{equation}
Therefore, to characterize the AoII's evolution, it is sufficient to report the evolution of the system's state $S_t$. To that end, let $\psi_t$ denote the action taken at time $t$, where $\psi_t=1$ if a transmission is initiated and $0$ otherwise. Given the available actions that the transmitter can take and the possible transitions of the process $\big(d_t\big)_{t\in\mathbb{N}}$, it is essential to characterize the relationship between $S_{t+1}$ and $S_t$. To that end, we distinguish between two cases: \\
$\bullet$ \textbf{Case} $\boldsymbol{1}$ - $S_t=0$: In this case, $d_t=0$. Let us now assume that the transmitter decides to remain idle for the duration of the timeslot $t$. At the next timeslot $t+1$, we could end up in one of the following situations: 1) either $X_{t+1}$ will jump in value and $d_{t+1}$ becomes equal to $1$, or 2) $d_{t+1}$ remains equal to $0$. As per our adopted Markovian model for $d_t$, these two events happen with a probability $1-\alpha$ and $\alpha$ respectively. To that end, we obtain
\begin{align}
&\Pr\big(S_{t+1}=0|S_t=0,\psi_t=0\big)=\alpha,\nonumber\\ &\Pr\big(S_{t+1}=1|S_t=0,\psi_t=0\big)=1-\alpha.
\label{notransmissionequal0}
\end{align}
Let us now consider that the transmitter proceeds with a transmission at time $t$. This transmission will have no effect on $d_t$, as $d_t$ is already equal to $0$. In other words, we will still remain in the GOOD regime zone as was reported in Fig. \ref{sourcechain1111}. Now, given that the transitions of the process $(d_t)_{t\in\mathbb{N}}$ does not depend on the exact difference $|X_t-\hat{X}_t|$, we can conclude that regardless of the channel realization, we have
\begin{align}
&\Pr\big(S_{t+1}=0|S_t=0,\psi_t=1\big)=\alpha,\nonumber\\ &\Pr\big(S_{t+1}=1|S_t=0,\psi_t=1\big)=1-\alpha.
\label{transmissionequal0}
\end{align}
$\bullet$ \textbf{Case} $\boldsymbol{2}$ - $S_t\neq 0$: In this case, we have $d_t=1$. Let us now consider that the transmitter opted out from any transmission at time $t$. At the next timeslot $t+1$, we may end up in one of the following situations: 1) either $d_{t+1}$ will remain equal to $1$, or 2) $X_{t+1}$ jumps in value and we go back to the GOOD regime reported in Fig. \ref{sourcechain1111}. As per our adopted Markovian model for $d_t$, these two events happen with a probability $\beta$ and $1-\beta$ respectively. To that end, we obtain
\begin{align}
&\Pr\big(S_{t+1}=S_t+1|S_t\neq0,\psi_t=0\big)=\beta,\nonumber\\ & \Pr\big(S_{t+1}=0|S_t\neq0,\psi_t=0\big)=1-\beta.
\label{notransmissionnotequal0}
\end{align}
\color{black}
Let us now consider that the transmitter decides to transmit a status update to the monitor at time $t$. By taking into account the possible channel realizations, we distinguish between two cases:
\begin{itemize}
\item $h_t=0$: In this case, the packet is not successfully delivered to the monitor. Accordingly, from the monitor's perspective, this is similar to the case where no transmission is initiated. Therefore, the evolution of $S_t$ follows the transitions reported in (\ref{notransmissionnotequal0}). 
\item $h_t=1$: We recall that during the transmission time, the value of the information process may change. To that end, with probability $1-\beta$, $X_t$ would have jumped values, and the transmitted information became obsolete. In fact, in this case, by the time the receiver gets the information, the information has already changed and the error continues. Consequently, we have
\begin{align}
&\Pr\big(S_{t+1}=S_t+1|S_t\neq0,\psi_t=1,h_t=1\big)=1-\beta,\nonumber\\ &\Pr\big(S_{t+1}=0|S_t\neq0,\psi_t=1,h_t=1\big)=\beta.
\label{transmissionnotequal0}
\end{align}
\end{itemize}
By taking into account the independence between the transitions of the process $\big(d_t\big)_{t\in\mathbb{N}}$ and the channel realizations, we can summarize the transitions of $S_t$ as follows
\begin{align}
&\Pr\big(S_{t+1}=S_t+1|S_t\neq0,\psi_t=1\big)=p_f\beta+(1-\beta)p_s\triangleq a,\nonumber \\ & \Pr\big(S_{t+1}=0|S_t\neq0,\psi_t=1\big)=p_f(1-\beta)+p_s\beta=1-a.
\label{transmissionnotequal0total}
\end{align}
\color{black}
Given the above system's dynamics, one can notice a necessity to impose some restrictions on the parameters and functions involved. Effectively, for packet transmission to be useful to the system's performance, we need to have
\begin{equation}
\Pr\big(S_{t+1}=0|S_t\neq0,\psi_t=1\big)>\Pr\big(S_{t+1}=0|S_t\neq0,\psi_t=0\big).
\label{equationconditionparameter}
\end{equation}
If this condition is violated, then transmitting a packet does not improve the system's overall performance. Specifically, this means that the information process changes drastically at each timeslot to the point that if we transmit a packet, the packet becomes obsolete by the time it arrives at the receiver. From (\ref{equationconditionparameter}), we can conclude that the condition is equivalent to having $a<\beta$. Next, let us consider that a packet is transmitted at each timeslot. Given the dynamics of the system, we have $\Pr\big(S_{t+1}=S_t+1|S_t\neq0,\psi_t=1\big)=a$. In other words, even if a packet is transmitted at every timeslot, there is still a chance for the system's penalty to grow. 
To prevent the situation where even a transmission at each timeslot will still lead to an unbounded penalty, it is necessary to impose the following condition
\begin{equation}
\sum\limits_{k=0}^{+\infty}f(k)a^{k}<+\infty.
\label{conditiononf}
\end{equation}
Note that, for similar reasons, analogous conditions have been previously adopted in the AoI framework for communication over unreliable channels \cite{8845106}. With the system's evolution clarified, we can now formulate our problem and find its optimal solution.
\subsection{Problem Formulation}
Let $\pi$ represents a transmission policy that determines the packets being sent over time. The transmission policy $\pi$ is defined as a sequence of actions $\pi=(\psi^{\pi}_0,\psi^{\pi}_1,\ldots)$. Let $\Pi$ denotes the
set of all \emph{causal} scheduling policies, i.e., where the decisions are taken without any knowledge of the future. Our optimization problem can be formulated as follows
\begin{equation}
\setlength{\belowdisplayskip}{0pt} \setlength{\belowdisplayshortskip}{0pt}
\setlength{\abovedisplayskip}{0pt} \setlength{\abovedisplayshortskip}{0pt} 
\begin{aligned}
& \underset{\pi\in \Pi}{\text{minimize}}
& & J_{\pi}(S_0)\triangleq\limsup_{T\to+\infty}\:\frac{1}{T}\mathbb{E}^{\pi}\Big(\sum_{t=0}^{T-1}f(S^{\pi}_t)|S_0\Big),\\
& \text{subject to}
& & C_{\pi}(S_0)\triangleq\limsup_{T\to+\infty}\:\frac{1}{T}\mathbb{E}^{\pi}\Big(\sum_{t=0}^{T-1}\psi^{\pi}_t|S_0\Big)\leq \delta,
\end{aligned}
\label{originalobjective}
\end{equation}
where $f: [0,+\infty)\mapsto [0,+\infty)$ is a non-decreasing function of $S^{\pi}_t$, and $0<\delta\leq 1$ is the highest update rate allowed. The above problem belongs to the family of Constrained Markov Decision Process (\textbf{CMDP}), which are known to be generally challenging to solve optimally. To address these challenges, we proceed in the sequel with a Lagrange approach and provide a step-by-step analysis to solve problem (\ref{originalobjective}) optimally.
\section{Problem Optimization}
\label{optimizationgeneralklshi}
\subsection{Lagrange Approach}
\label{lagrangeapproachhh}
The Lagrange approach consists of transforming the constrained problem (\ref{originalobjective}) to an unconstrained one by incorporating the constraint in the objective function. Specifically, let us introduce the Lagrange multiplier $\lambda\in\mathbb{R}^{+}$. We define the Lagrangian function as
\begin{equation}
\mathcal{L}(\lambda,\pi)=\limsup_{T\to+\infty}\:\frac{1}{T}\mathbb{E}^{\pi}\Big(\sum_{t=0}^{T-1}f(S^{\pi}_t)+\lambda\psi^{\pi}_t|S_0\Big)-\lambda\delta.
\label{lagrangeobjective}
\end{equation}
Given that $\lambda\geq 0$, it can be regarded as a penalty that is paid for a packet transmission. Ideally, we would like to find a certain $\lambda^*$ for which minimizing the function (\ref{lagrangeobjective}) across all policies $\Pi$ allows us to derive the optimal policy of the constrained problem (\ref{originalobjective}). To proceed in that direction, let us consider the following optimization problem
\begin{equation}
\underset{\pi\in \Pi}{\text{min}}\:\:\mathcal{L}(\lambda,\pi),
\label{firstobjective}
\end{equation}
for any fixed $\lambda\in\mathbb{R}^+$. Knowing that $\lambda\delta$ is independent of the chosen policy $\pi$, the above minimization problem is equivalent to the following
\begin{equation}
\underset{\pi\in \Pi}{\text{min}}\:\:h(\lambda,\pi)=\underset{\pi\in \Pi}{\text{min}}\limsup_{T\to+\infty}\:\frac{1}{T}\mathbb{E}^{\pi}\Big(\sum_{t=0}^{T-1}f(S^{\pi}_t)+\lambda\psi^{\pi}_t|S_0\Big).
\label{secondobjective}
\end{equation}
Therefore, we focus on the optimization problem (\ref{secondobjective}). Based on the system's dynamics previously detailed in Section \ref{penaltydynamics}, the above problem can be cast into an infinite horizon average cost Markov Decision Process (\textbf{MDP}) as follows
\begin{itemize}
\item \textbf{States}: The state of the system $S_t$ coincides with that reported in Section \ref{penaltydynamics}. Accordingly, the state space of interest $\mathbb{S}$ is the space of natural numbers $\mathbb{N}$. 
\item \textbf{Actions}: At any time $t$, the possible actions that can be taken by the transmitter are to either initiate a new transmission ($\psi_t=1$) or to stay idle ($\psi_t=0$).
\item \textbf{Transitions probabilities}: The transitions probabilities between the different states correspond to those previously reported in Section \ref{penaltydynamics}.
\item \textbf{Cost}: Given the objective function of the problem, the instantaneous cost is set to $C(S_t,\psi_t)=f(S_t)+\lambda\psi_t$.
\end{itemize}
To obtain the optimal policy of an infinite horizon average cost MDP, if it exists, it is well-known that it is sufficient to solve the following Bellman equation \cite{Bertsekas:2000:DPO:517430}
\begin{equation}
\theta + V(S)=\min_{\psi\in\{0,1\}}\big\{f(S)+\lambda\psi+\sum_{S'\in\mathbb{S} }\Pr(S\rightarrow S'|\psi)V(S')\big\},
\label{generalbellmanequation}
\end{equation}
where $\Pr(S\rightarrow S'|\psi)$ is the transition probability from state $S$ to $S'$ given the action $\psi$, $\theta$ is the optimal value of (\ref{secondobjective}), and $V(S)$ is the differential cost-to-go function. However, this is notoriously known to be a challenging task \cite{Bertsekas:2000:DPO:517430}. We leverage our system's particularity to circumvent these challenges and provide key structural results on the value function $V(\cdot)$. Using these results, we proceed to solve the Bellman equation, as will be seen in the sequel.

%
%
%
%
\subsection{Structural Results}
As previously explained, we start by studying the particularity of the value function. Before doing so, we first distinguish between two types of functions $f(S)$ based on their behavior for large $S$. To that end, we define
\begin{itemize}
\item Unbounded $f(\cdot)$: In this case, the function $f(\cdot)$ grows indefinitely with the system's state
\begin{equation}
\lim_{S \to +\infty} f(S)=+\infty.
\end{equation}
The list of such functions includes the linear function reported in (\ref{functional1}).
\item Bounded $f(\cdot)$: In this case, the penalty of the system saturates and reaches a fixed limit
\begin{equation}
\lim_{S \to +\infty} f(S)=L>0.
\end{equation}
An example that belongs to this family of functions is the time-threshold function reported in (\ref{functional3}).
\end{itemize}
To analyze the bounded function case, we will proceed with a truncation of the state space $\mathbb{S}=\mathbb{N}$. Specifically, from the limit definition, we have
\begin{equation}
\forall \epsilon>0,\:\exists S_{\textnormal{thresh}}: \forall S\geq S_{\textnormal{thresh}}, |f(S)-L|<\epsilon.
\end{equation}
Accordingly, we can choose an arbitrarily small $\epsilon$ such that  $f(S)\approx f(S_{\textnormal{thresh}}), \:\: \forall S\geq S_{\textnormal{thresh}}$. To that end, we let $\mathbb{S}=\{0,1,\ldots,S_{\textnormal{thresh}}\}\subseteq\mathbb{N}$. Although this truncation will have a negligible effect on the performance for a small $\epsilon$, it will prove to have analytical benefits in deriving the optimal transmission policy. With this distinction in mind, we lay out the following lemma. 
\begin{lemma}[Non-decreasing Property of $V(\cdot)$]
For both function classes, the differential cost-to-go function $V(S)$ is a non-decreasing function of $S$.
\label{nondecreasingproperty}
\end{lemma}
\begin{proof}
The proof is in Appendix \ref{proofoflemmanondecreasing}.
\end{proof}
\noindent Next, we leverage the above lemma to establish the fundamental proposition below. 
\begin{proposition}[Structure of the Optimal policy]
For any $\lambda\in\mathbb{R}^{+}$, and for both function classes, the transmission policy that optimally solves problem (\ref{secondobjective}) is a threshold policy.
\label{thresholdproposition}
\end{proposition}
\begin{proof}
The proof is in Appendix \ref{proofthresholdproposition}. 
\end{proof}
\noindent The above proposition allows us to have a road-map to solve the Bellman equation. Knowing that a threshold policy is optimal, we restrict our attention to this class of policies to simplify and solve the Bellman equation. Consequently, we lay out the following theorem.


\begin{theorem}[Optimal Policy]
The optimal transmission policy $\pi^*_\lambda$ can be summarized as follows
\begin{itemize}
\item Unbounded $f(\cdot)$: $\pi^*_\lambda$ is a threshold policy such that a transmission is initiated when $S_t\geq n^*_\lambda$ where
\begin{align}
n^*_\lambda=\inf\{&n\in\mathbb{N}^*:H(n)>0\}-1,
\end{align}
and $H(n)$ and $\theta_n$ are equal to
\begin{equation}
H(n)=\frac{-\theta_n(\beta-a)+\lambda(\beta-1)}{(1-a)(\beta-a)}+\sum\limits_{k=n}^{+\infty}f(k)a^{k-n},
\end{equation}
\begin{equation}
\theta_n=\frac{\frac{f(0)}{1-\alpha}+\sum\limits_{k=1}^{n-1}f(n-k)\beta^{n-k-1}+\beta^{n-1}\sum\limits_{k=n}^{+\infty}f(k)a^{k-n}+\frac{\lambda \beta^{n-1}}{1-a}}{\frac{1}{1-\alpha}+\frac{1-\beta^{n-1}}{1-\beta}+\frac{\beta^{n-1}}{1-a}}.
\end{equation}
\item Bounded $f(\cdot)$: $\pi^*_\lambda$ is a threshold policy if 
\begin{equation}
\lambda< \frac{(\beta-a)(f(S_{\textnormal{thresh}})-\theta_{S_{\textnormal{thresh}}})}{1-\beta},
\end{equation}
and the optimal threshold $n^*_\lambda$ is equal to
\begin{equation}
n^*_\lambda=\inf\{n\in\mathbb{S}\setminus\{0\}:H'(n)>0\}-1,
\label{optimalthresholdbounded}
\end{equation}
where $H'(n)$ and $\theta'_n$ are reported in Table \ref{allequationsfortheorem}. Otherwise, the optimal transmission policy $\pi^*_\lambda$ is to never transmit, and we set $n^*_\lambda=S_{\textnormal{thresh}}+1$. 
\end{itemize}
\label{maintheorem}
\end{theorem}
\begin{table*}[ht]
\centering
\begin{tabular}{p{0.99\linewidth}}
\begin{equation*}
H'(n)=\frac{-\theta'_n(\beta-a)+\lambda(\beta-1)}{(1-a)(\beta-a)}+\sum\limits_{k=n}^{S_{\textnormal{thresh}}-1}f(k)a^{k-n}+\frac{a^{S_{\textnormal{thresh}}-S}f(S_{\textnormal{thresh}})}{1-a}
\end{equation*}
\begin{equation*}
\theta'_n=\frac{\frac{f(0)}{1-\alpha}+\sum\limits_{k=1}^{n-1}f(n-k)\beta^{n-k-1}+\beta^{n-1}\sum\limits_{k=n}^{S_{\textnormal{thresh}}-1}f(k)a^{k-n}+\frac{\lambda \beta^{n-1}}{1-a}+\beta^{n-1}\frac{a^{S_{\textnormal{thresh}}-S}f(S_{\textnormal{thresh}})}{1-a}}{\frac{1}{1-\alpha}+\frac{1-\beta^{n-1}}{1-\beta}+\frac{\beta^{n-1}}{1-a}}
\end{equation*}
\\
\hline
\caption{Expressions of $H'(n)$ and $\theta'_n$ in the bounded function case.}
\label{allequationsfortheorem}
\end{tabular}
\vspace{-28pt}
\end{table*}
\begin{proof}
The proof is in Appendix \ref{proofofmaintheorem}.
\end{proof}
\noindent The next step consists of deriving a closed-form expression of $C_{\pi^*_{\lambda}}$ for any $\lambda\in\mathbb{R}^+$. Finding this expression will allow us to propose an iterative algorithm later on that finds the optimal transmission policy, as will be seen in Section \ref{algoimplementation}. To that end, we provide the following proposition.
\begin{proposition}[Update Rate] The average update rate of the transmission policy $\pi^*_\lambda$ is
\begin{itemize}
\item Unbounded $f(.)$:
\begin{equation}
C_{\pi^*_{\lambda}}=\begin{cases}
\frac{(1-\alpha)\beta^{n^*_{\lambda}-1}}{(1-a)(1+\frac{(1-\alpha)(1-\beta^{n^*_{\lambda}})}{1-\beta}+\frac{(1-\alpha)a\beta^{n^*_{\lambda}-1}}{1-a})}, \: &\textnormal{if} \:\: n^*_{\lambda}\in\mathbb{N}^*,\\
1, \: &\textnormal{if} \:\: n^*_{\lambda}=0.\\
\end{cases}
\end{equation}
\item Bounded $f(.)$: it coincides with the unbounded case expression for any $n^*_{\lambda}\in\mathbb{S}$, and is equal to $0$ if $n^*_{\lambda}=S_{\textnormal{thresh}}+1$. 
\end{itemize}
\label{propositionupdaterate}
\end{proposition}
\begin{proof}
The proof is in Appendix \ref{proofpropositionupdaterate}.
\end{proof}
\subsection{Optimal Policy}
Thus far, we have focused on finding the optimal transmission policy $\pi^*_\lambda$ that solves problem (\ref{secondobjective}), which it turns solves (\ref{firstobjective}). However, our primary goal remains to optimally solve the original constrained problem reported in (\ref{originalobjective}). It turns out, we can relate the optimal policy for the constrained problem to that of (\ref{firstobjective}) if certain conditions are satisfied. To that end, let us first define $\lambda^*\triangleq\inf\{\lambda\in\mathbb{R}^+: C_{\pi^*_{\lambda^*}}\leq\delta\}$ and $\vartheta=\frac{1-\alpha}{2-\alpha-a}$. With these definitions in mind, we summarize our findings in the following theorem.
\begin{theorem}[Optimal Policy of the Constrained Problem]
The optimal transmission policy of problem (\ref{originalobjective}) can be summarized as follows
\begin{itemize}
\item Unbounded $f(\cdot)$: the optimal transmission policy 
is a randomized threshold policy with parameter $\mu^*$ such that
\begin{itemize}
\item The thresholds $n^*_{\lambda^*}-1$ and $n^*_{\lambda^*}$ are adopted with probability $\mu^*$ and $1-\mu^*$ respectively.
\item $\mu^*$ is chosen to ensure that the randomized policy has an average update rate equal to $\delta$. In other words, 
\begin{equation}
\mu^*=\frac{\delta-C_{\pi^*_{\lambda^*,2}}}{C_{\pi^*_{\lambda^*,1}}-C_{\pi^*_{\lambda^*,2}}},
\end{equation}
where $C_{\pi^*_{\lambda^*,1}}$ and $C_{\pi^*_{\lambda^*,2}}$ are the average update rate when the thresholds $n^*_{\lambda^*}-1$ and $n^*_{\lambda^*}$ are used respectively. 
\end{itemize}
\item Bounded $f(\cdot)$: the optimal transmission policy coincides with the unbounded function case if $\delta<\vartheta$. Otherwise, an optimal transmission policy is to transmit a packet in every timeslot $t$ where $S_t\neq0$.
\end{itemize}
\label{theoremrelatingtoconstrained}
\end{theorem}
\begin{proof}
The proof is in Appendix \ref{prooftheoremrelatingtoconstrained}.
\end{proof}

\subsection{Algorithm Implementation}
\label{algoimplementation}
To obtain the optimal policy, we implement a specific low-complexity algorithm, as explained below. The first step in our algorithm implementation consists of finding the optimal threshold for any fixed $\lambda\in\mathbb{R}^+$. To that end, we recall from our analysis in the proof of Theorem \ref{maintheorem}, the functions $H(n)$ and $H'(n)$ are both non-decreasing with $n$. Accordingly, we can use the binary search algorithm \cite{10.5555/1614191} to find the optimal threshold for any $\lambda$. Specifically, starting from an initial interval $I=[1,2]$, we exponentially enlarge this interval as long as $n^*_{\lambda}\not\in I$. When the interval is large enough to contain $n^*_{\lambda}$, a binary search algorithm is adopted to find it.
Interestingly, this whole procedure is computationally efficient as it requires at most ${\displaystyle O(\log n^*_{\lambda})}$ iterations.\color{black}
\begin{remark}[Implementation Considerations] The evaluation of $H(n)$ requires the calculation of an infinite sum series. However, given the assumption in (\ref{conditiononf}), we have
\begin{equation}
\forall \epsilon>0,\:\exists M: \forall m\geq M, |\sum\limits_{k=n}^{m}f(k)a^{k-n}-\sum\limits_{k=n}^{+\infty}f(k)a^{k-n}|<\epsilon.
\end{equation}
Accordingly, we can always consider a finite sum satisfying a predefined precision criterion. 
\end{remark}
Next, we derive a scheme to find the optimal Lagrange multiplier $\lambda^*$. To that end, we note that the average update rate $C_{\pi^*_{\lambda}}$ is non-increasing with $\lambda$ \cite{sennott1993,BEUTLER1985236} and that $C_{\pi^*_{0}}=1$. Accordingly, we employ a bisection search method to find $\lambda^*$\cite{10.5555/1614191}, which also has low complexity. Specifically, as it was done for the binary search algorithm, we start with an initial interval $I_0=[\lambda^0_{\textnormal{min}},\lambda^0_{\textnormal{max}}]$ where $\lambda^0_{\textnormal{min}}=0$ and $\lambda^0_{\textnormal{max}}=1$. As long as $C_{\pi^*_{\lambda^t_{\textnormal{max}}}}>\delta$, we set $\lambda^{t+1}_{\textnormal{min}}=\lambda^t_{\textnormal{max}}$ and $\lambda^{t+1}_{\textnormal{max}}=2\lambda^t_{\textnormal{max}}$. We do so until we end up with an interval $I_{t^*}=[\lambda^{t^*}_{\textnormal{min}},\lambda^{t^*}_{\textnormal{max}}]$ such that 
$C_{\pi^*_{\lambda^{t^*}_{\textnormal{min}}}}>\delta$ and $C_{\pi^*_{\lambda^{t^*}_{\textnormal{max}}}}\leq\delta$ for some $t^*\geq0$. This interval expansion finishes in at most ${\displaystyle O(\log\lambda^*)}$ steps. Given that at each step of the expansion an optimal threshold needs to be found, the overall complexity of this step is hence ${\displaystyle O(\log(\lambda^*)\times[\max_{0\leq t\leq \log(\lambda^*) } \log_2(n^*_{\lambda^t_{\textnormal{max}}})])}.$ \color{black}
The next step consists of evaluating the middle point of the interval $\xi_t=\frac{\lambda^t_{\textnormal{min}}+\lambda^t_{\textnormal{max}}}{2}$. If $C_{\pi^*_{\xi_t}}>\delta$, then we set $I_{t+1}$ to $[\xi_t,\lambda^t_{\textnormal{max}}]$. Otherwise, we set it to $[\lambda^t_{\textnormal{min}},\xi_t]$. We keep doing this until a convergence criterion is satisfied and the algorithm outputs $\xi_{\infty}$. The complexity of this step is ${\displaystyle O (\log(\frac{\lambda^{t^*}_{\textnormal{max}}-\lambda^{t^*}_{\textnormal{min}}}{\epsilon})\times[\max_{t\leq \log(\frac{\lambda^{t^*}_{\textnormal{max}}-\lambda^{t^*}_{\textnormal{min}}}{\epsilon})} \log(n^*_{\xi_t})])}$, where $\epsilon$ is the convergence tolerance. \color{black}Consequently, to get the optimal transmission policy, it is sufficient to set $n^*_{\lambda^*}$ of Theorem \ref{theoremrelatingtoconstrained} to $n^*_{\xi_{\infty}}$. Finally, the randomization parameter $\mu^*$ can be easily concluded using the resulting $n^*_{\lambda^*}$ that we adopt. A pseudo-code of the algorithm is reported in Appendix \ref{algorithmpseudocode}. 
\section{Comparison with the Error Framework}
\label{comparisonwitherror}
Given that the proposed AoII framework incorporates both an information aspect (through the function $g(\cdot,\cdot)$) and a time-aspect (through the function $f(\cdot)$), an interesting question is how such a framework compares to the standard error-based measure approach?\color{black}. This section answers this question by comparing the performance of both the error-optimal policy and the AoII-optimal policies. Interestingly, we can obtain the error-optimal transmission policy $\pi^*_e$ by adopting the following function $f_{\textnormal{error}}(S_t)=1$ if $S_t\neq0$ and $f(0)=0$, and applying Theorem \ref{theoremrelatingtoconstrained}. By doing so, we minimize the long-term average of the error measure $d_t$ depicted in Section \ref{modellll}. Let us now consider a simple setting where the communication goal can be written as $f_1(S_t)=S_t$, and let $\pi^*_a$ denote the corresponding AoII-optimal policy. Moreover, suppose that $\alpha=0.2$, $\beta=0.9$ and $p_s=0.8$. We compare the two policies in terms of average AoII (i.e., communication goal utility) and average error below. We also report the optimal thresholds for each policy. 
\begin{center}
\begin{tabular}{|c|c|c|c|c|c|c|}
 \hline
 $\delta$ & $\text{AoII}_{\pi^*_a}$ & $\text{AoII}_{\pi^*_e}$  & $\text{Error}_{\pi^*_a}$ & $\text{Error}_{\pi^*_e}$ & $n_{\pi^*_a}^*$  & $n_{\pi^*_e}^*$\\
  \hline
 $0.05$  & $4.5$ & $8.1$ & $0.85$ & $0.85$  & $13$ & $1$ \\
 $0.1$  & $3.1$ & $7.4$ & $0.8$ & $0.8$ & $8$ & $1$\\
 $0.4$  & $1$ & $2.5$ & $0.6$ & $0.6$ & $2$ & $1$\\
 \hline
\end{tabular}
 \captionof{table}{Performance comparison between $\pi^*_a$ and $\pi^*_e$.}
 \label{RRCSMAasdatryy}
\end{center}
\color{black}
Interestingly, the AoII-optimal policy achieves the same error performance as the error-optimal policy (i.e., the AoII-optimal policy is also error-optimal). On the other hand, the error-optimal approach is not AoII-optimal, as seen by the two policies' performance gap. 
This was first observed numerically in the work of Clement et al. \cite{9162726} and our work here provides a rigorous understanding of this phenomena. To see this more clearly, we recall that $f_{\textnormal{error}}(\cdot)$ is bounded. With this in mind, we distinguish between two cases
\begin{itemize}
\item $\delta>\vartheta=\frac{1-\alpha}{2-\alpha-a}$: In this case, the error-optimal policy is to send an update whenever $S_t\neq0$ (Theorem $2$). In other words, the update rate is not restrictive, and one can send an update whenever $S_t\neq0$ without violating it. It is easy to see that the error-optimal and AoII-optimal policies will coincide in this case.
\item $\delta<\vartheta=\frac{1-\alpha}{2-\alpha-a}$: In this more interesting case, we can see that the error-optimal policy consists in sending a packet whenever $S_t\neq0$ with a probability $1-\mu^*$. By doing so, the update rate constrained is satisfied with equality. On the other hand, the AoII-optimal policy is more elaborate than this. In fact, the AoII-optimal policy will depend on the instantaneous value of $S_t$, not on just whether or not it is equal to $0$ or $1$. Specifically, as seen in Theorem $2$, the AoII-optimal policy will alternate between two thresholds $n^*_{\lambda^*}-1$ and $n^*_{\lambda^*}$ in a way to satisfy the update rate with equality. Given that the AoII-optimal policy will only send packets when $S_t\neq0$ while satisfying the update rate constraint with equality, we can conclude that it is also error-optimal.
\end{itemize}
Therefore, one can see that the AoII-optimal policy is more elaborate and that it is error-optimal. With that in mind, we can lay out the following conclusion. 
\begin{conclusion}
Adopting AoII-optimal policies minimizes the average error while also helping achieve the communication's goal. On the contrary, the converse is not necessarily true.
\end{conclusion}
\color{black}
\section{Numerical Results}
\begin{figure*}[ht]
\centering
\begin{subfigure}{0.33\textwidth}
  \centering
  \includegraphics[width=.9\linewidth]{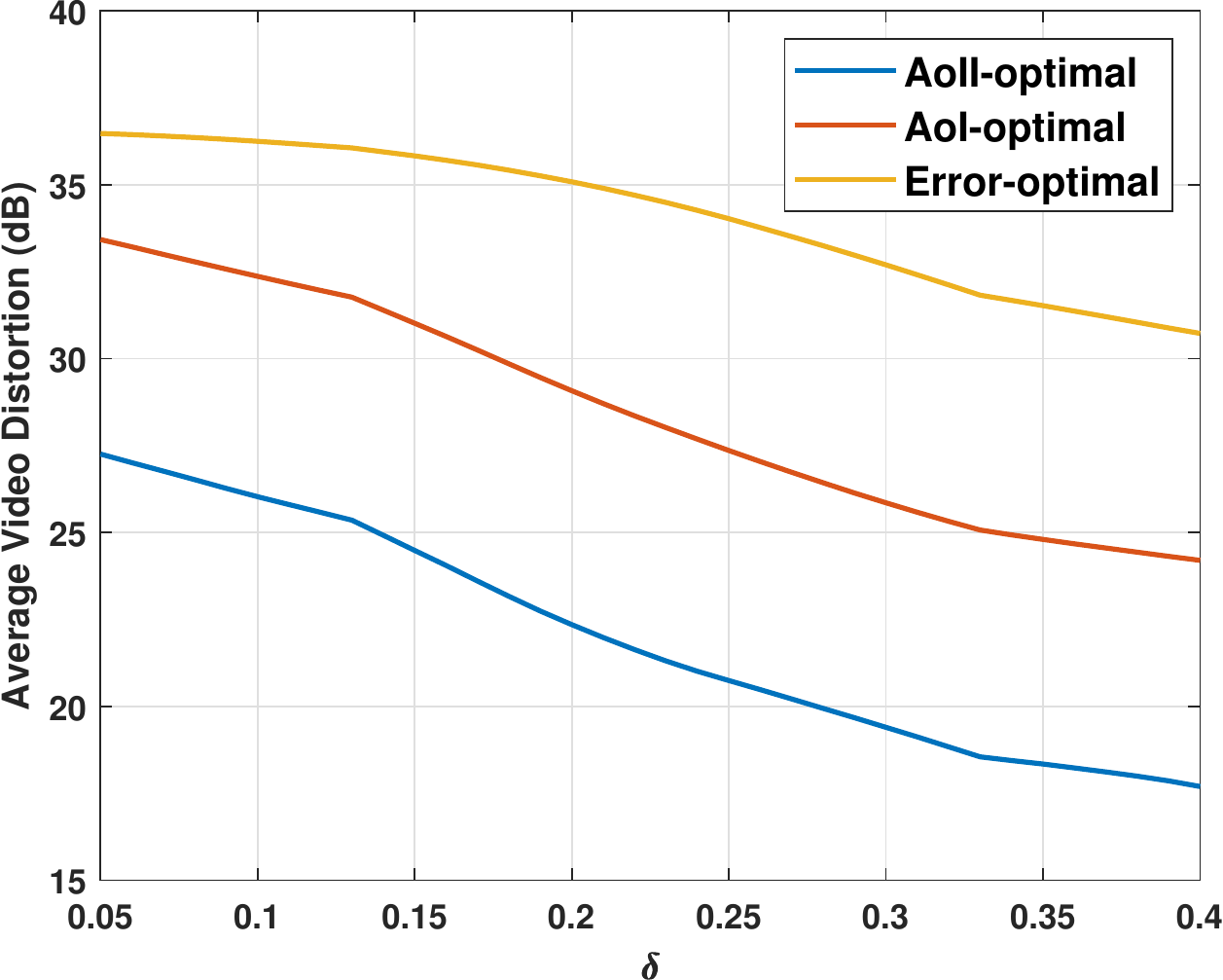}
  \caption{Video streaming.}
    \label{videocodingsimulations}
\end{subfigure}%
\begin{subfigure}{0.33\textwidth}
\centering
  \includegraphics[width=.9\linewidth]{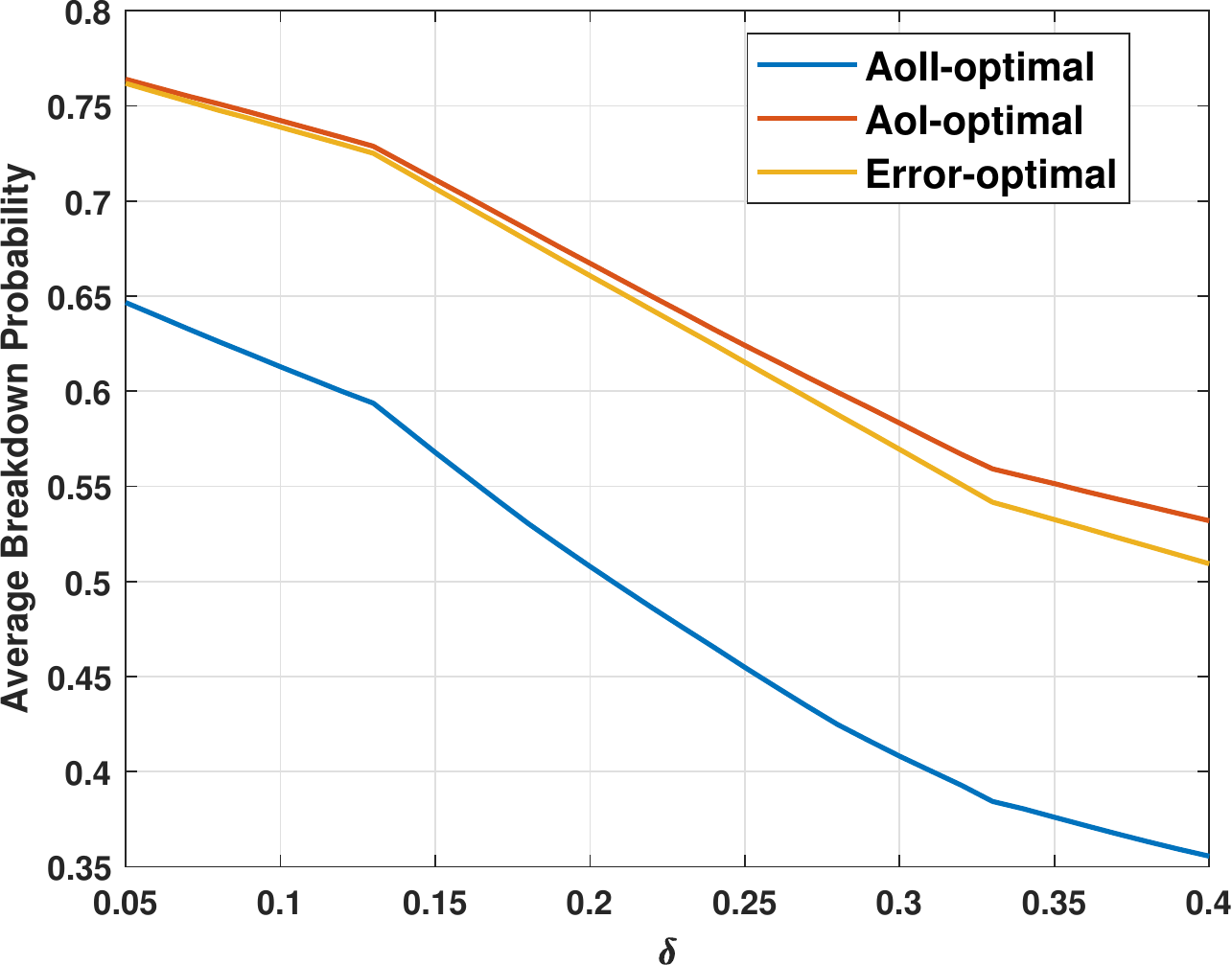}
  \caption{Machine overheating.}
\label{machinesimulations}
\end{subfigure}%
\begin{subfigure}{0.33\textwidth}
\centering
  \includegraphics[width=.9\linewidth]{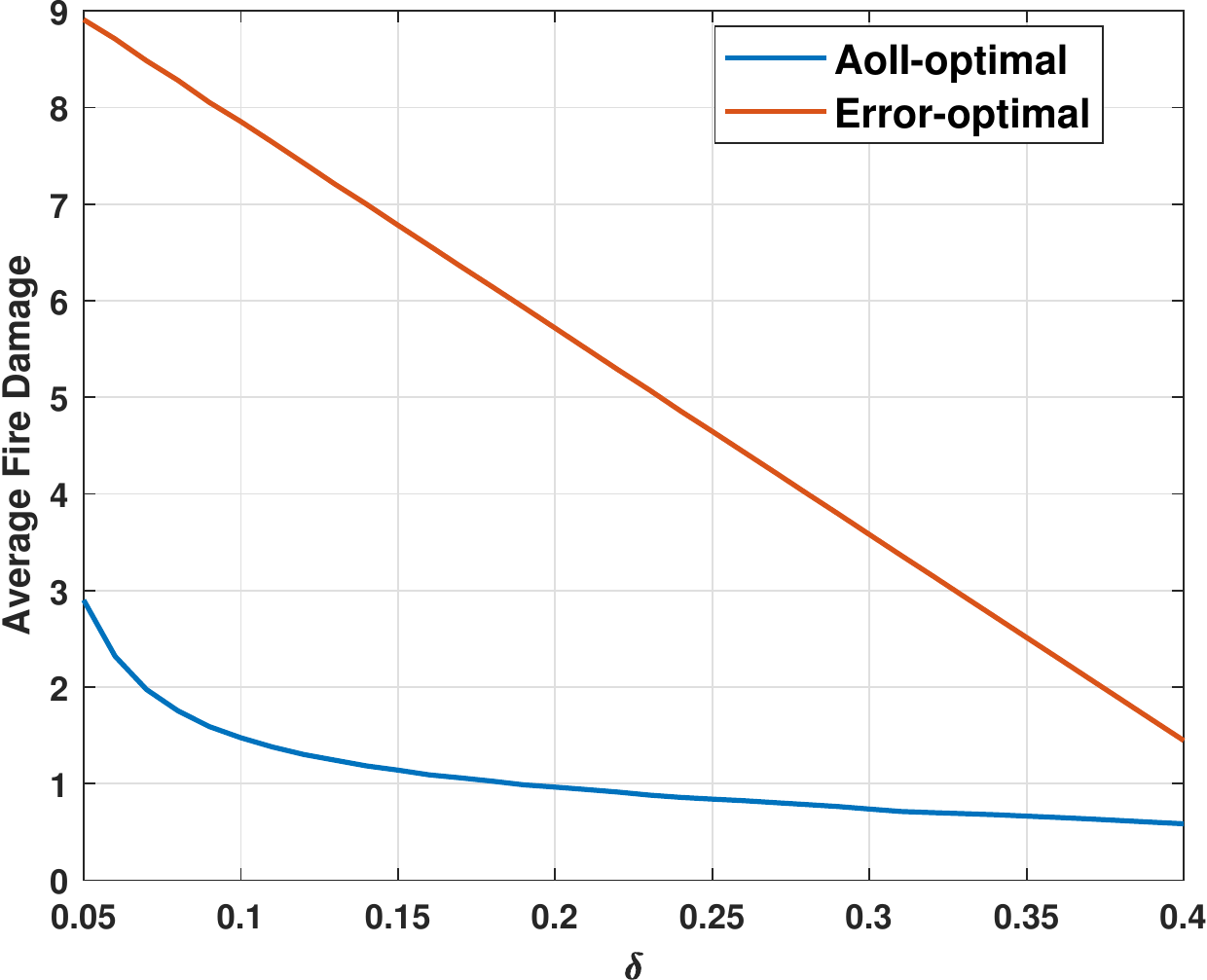}
  \caption{Fire monitoring.}
\label{firestuff}
\end{subfigure}%
\vspace{-5pt}
\caption{Illustrations of different simulations.}
\vspace{-5pt}
\label{simulationss}
\end{figure*}
In this section, we provide real-life applications of the AoII and compare the performance between the AoII-optimal, the AoI-optimal \cite{8377368}, and the error-optimal schemes. 
\label{numericalresults}
\subsection{Video Streaming}
We consider a transmitter-receiver pair where real-time video stream packets encoded using the standard MPEG-4 AVC (Advanced Video Coding) scheme are sent from one end to the other \color{black}. Time is slotted and normalized to the slot duration (i.e., the slot duration is taken as $1$). The video stream comprises frames, each of which is a 1-D vector of length $M$ in line-scan order. The stream's total duration is $T$ timeslots. At each timeslot, a frame of the video stream is sent by the transmitter side. We suppose that the channel at timeslot $t$ is $X_t$, and its estimate at the transmitter's side is $\hat{X}_t$. The transmitter can send pilot signals and learn the channel at the beginning of each timeslot. However, this training succeeds with a probability $0<p_s<1$ and incurs a cost, knowing that an average cost budget $\delta$ cannot be surpassed. As explained in Section \ref{modellll}, by adopting a Markovian channel model, it can be shown that without channel learning, the process $d_t$ can be modeled using a Markov chain. We suppose that this chain's parameters are $\alpha$ and $\beta$ as previously depicted in earlier sections. We assume that the receiver successfully decodes packets if $d_t=0$ and a transmission error occurs otherwise. At the receiver, we assume a simple loss concealment scheme where the lost frame due to a transmission error is replaced by the previous frame. The error propagation process is modeled with a geometric attenuation factor resulting from spatial filtering. Let us assume that each error introduces an initial error power $\gamma$, and the cross-correlation factor between each successive error is $\rho$. By following the derivations in \cite{4490280}, we can show that the video distortion model provided in \cite{4490280} is a special case of the AoII where $f(0)=0$ and for any $S_t=S>0$, we have \color{black}
\begin{equation}
f(S)=\gamma S(\alpha_0+(S-1)(\tau+\rho(S-1)+c\rho(S-2))),
\end{equation}
where $\tau=1+\alpha_0\rho+c$, and $(\alpha_0,c)$ are two parameters of the video stream. It is important to note that the \textbf{channel training goal} is to minimize the total average distortion of the receiver's video signal. We are not interested in having fresh estimates of $X_t$ (AoI metric) or minimizing the channel prediction error (standard error metric). Therefore, we can see how by tweaking the function $f(.)$, the AoII allows us to capture the channel training's goal. To highlight our AoII approach's benefits, we compare it to the AoI and the standard error-based frameworks for this particular scenario. Specifically, we evaluate the average video distortion resulting from adopting the optimal policies for these different metrics. We consider $\alpha=0.5$, $\beta=0.8$, $p_s=0.8$, $T=10^6$, $\rho=0.8$, $c=2$, $\gamma=1$, and $\alpha_0=4$. As seen in Fig. \ref{videocodingsimulations}, the AoII-optimal policy outperforms the two other policies for any $\delta$.
\subsection{Machine Overheating}
In this scenario, we assume that a transmitter informs a remote monitor about whether or not the monitored electrical machine is overheating. An abnormal increase in temperature in electrical devices creates thermal stress on the machine, leading to the breakdown of the electrical insulation (e.g., motor winding insulation). This itself will lead to an eventual malfunction of the machine. Therefore, the transmitter needs to inform the monitor of the temperature's status and solicit instructions to minimize the probability of the machine malfunction. We suppose that the transmitter is limited on how often it can update the monitor. The average update rate allowed is $\delta$. We consider that $d_t=1$ when the machine is overheating at time $t$ and $d_t=0$ otherwise. We also assume that $d_t$ evolves as a Markov chain, and its parameters are $\alpha$ and $\beta$. Following the study in \cite{feilat}, the probability of an insulation breakdown under temperature stress follows a Weibull distribution. Precisely,
\begin{equation}
\Pr(\textnormal{Breakdown time}<t)=1-\exp(-(t/\gamma)^{\rho}),
\end{equation}
where $\gamma$ and $\rho$ are parameters that depend on the machine's characteristics. When there is no temperature stress, the breakdown probability is negligible. The \textbf{communication goal} is to choose the update times such that the probability of a breakdown since the stress was applied is minimized. We can see that this probability is a special case of the AoII where $f(0)=0$, and for any $S_t=S>0$, we have $f(S)=1-\exp(-(S/\gamma)^{\rho})$. We evaluate the average breakdown probability that results from adopting the optimal policies for the standard three metrics. We consider $\alpha=0.2$, $\beta=0.9$, $p_s=0.8$, $\rho=1$, and $\gamma=1$. As seen in Fig. \ref{machinesimulations}, the AoII-optimal policy outperforms the two other policies for any $\delta$.
\subsection{Fire Monitoring}
Contrary to the previous cases, we consider an application outside the scope of traditional communication networks. Specifically, we consider a scenario where fires happen independently and fire stations have to respond to them. Accordingly, this application falls under the decision problems umbrella. As found by the UK fire research station, the spread of fire can be represented through an exponential statistical model \cite{ramachandran1986exponential}. Specifically, 
\begin{equation}
F(t)=\min\{F_{\textnormal{max}};F_{\textnormal{init}}\exp(\gamma(t-t_{\textnormal{fire}}))\},
\label{equationfire}
\end{equation}
where $F(t)$ is the amount of fire damage at time $t$ since ignition, $F_{\textnormal{init}}$ is the initial ignite damage, $\gamma$ is the fire growth parameter, $F_{\textnormal{max}}$ is the maximum possible damage, and $t_{\textnormal{fire}}$ is the ignition time. Given the restricted resources, the fire stations are limited on how often they can respond to fires, as an average response rate of $\delta$ cannot be surpassed. We consider that $d_t=1$ when a fire is happening at time $t$ and $d_t=0$ otherwise. We also assume that $d_t$ evolves as a Markov chain, and its parameters are $\alpha$ and $\beta$ such that $\beta=1$. The \textbf{goal} is to minimize the total average fire damage. Using (\ref{equationfire}), we can see that the fire damage is a special case of the AoII where $f(0)=0$, and for any $S_t=S>0$, we have $f(S)=\min\{F_{\textnormal{max}};F_{0}\exp(\gamma S)\}$. We evaluate the average fire damage that results from adopting the AoII-optimal and the error-optimal policies. We consider $\alpha=0.2$, $p_s=1$, $F_{\textnormal{max}}=10$, $\gamma=0.1$, and $F_{\textnormal{init}}=1$. As seen in Fig. \ref{firestuff}, the AoII-optimal policy outperforms the error approach for any $\delta$. This example shows that the AoII is not restricted to communication networks and can be utilized in various other frameworks. 
\vspace{-7pt}
\section{Conclusion and Future Work}
\label{conclusionsss}
In this paper, we have shown how the AoII metric enables semantics-empowered communication, where the communication's goal is taken into account. Additionally, we have developed an optimal transmission policy that minimizes the AoII, and we showcased its substantial performance advantages. Future research directions include the extension to more general information source models, examining continuous-time systems, investigating multi-user scenarios, and providing even a broader range of real-life applications of the AoII. Additionally, the implementation of the AoII in real-life environments and addressing the challenges that arise in deriving the optimal sampling policy and the estimation of the penalty functions $f(\cdot)$ and $g(\cdot,\cdot)$) through ML techniques, in this case, are to be considered. 
\color{black}. 
\vspace{-8pt}
\bibliographystyle{IEEEtran}
\bibliography{trialout}
\appendices

\section{Proof of Lemma \ref{nondecreasingproperty}}
\label{proofoflemmanondecreasing}
The first step consists of simplifying the Bellman equation. For the unbounded function class, and given the dynamics of the system reported in Section \ref{penaltydynamics}, we can rewrite the Bellman equation as follows
\begin{align}
\theta + V(0)=&\min\big\{f(0)+\alpha V(0)+(1-\alpha)V(1);f(0)+\lambda\nonumber\\
&+\alpha V(0)+(1-\alpha)V(1)\big\},\nonumber\\
\theta + V(S)=&\min\big\{f(S)+\beta V(S+1)+(1-\beta)V(0);f(S)\nonumber\\&+\lambda+aV(S+1)+(1-a)V(0)\big\}, \quad \forall S\in\mathbb{N}^*.
\label{bellmanequation}
\end{align}
Notice that the upper part of the minimization in (\ref{bellmanequation}) is associated with choosing $\psi=0$, i.e., letting the transmitter idle, and the lower part with $\psi=1$, i.e., initiating a transmission. 
To prove the desired results, we leverage the Relative Value Iteration Algorithm (\textbf{RVIA}) \cite{Bertsekas:2000:DPO:517430}. The RVIA is an iterative algorithm that calculates the differential cost-to-go function $V(S)$  of the Bellman equation reported in (\ref{generalbellmanequation}). To that end, and for any state $S\in\mathbb{N}$, let $V_t(S)$ designate the differential cost-to-go function estimate at iteration $t$. Also, let us denote by $T(V_t)(S)$ the mapping obtained by applying the right-hand side of the Bellman's equation
\begin{equation}
T(V_t)(S)=\min_{\psi\in\{0,1\}}\big\{f(S)+\lambda\psi+\sum_{S'\in\mathbb{N} }\Pr(S\rightarrow S'|\psi)V_t(S')\big\},
\end{equation}
where $\Pr(S\rightarrow S'|\psi)$ is the transition probability from state $S$ to $S'$ given the action $\psi$. Without loss of generality, we suppose that $V_0(S)=0$ for all states $S\in\mathbb{N}$ and we let $S=0$ be the reference point of the algorithm. With that in mind, the estimate of the differential cost-to-go function is updated as follows
\begin{equation}
V_{t+1}(S)=T(V_t)(S)-T(V_t)(0), \quad \forall S\in\mathbb{N}.
\end{equation}
Note that $V_t(0)=0$ holds for all iterations $t$. As stated in \cite[Proposition~3.1]{Bertsekas:2000:DPO:517430}, the above algorithm converges to the differential cost-to-go function $V(S)$ (i.e., $\lim_{t\to+\infty} V_{t}(S)=V(S), \:\:\forall S\in\mathbb{N}$). Accordingly, if we can show the non-decreasing property of $V_{t}(S)$ for any time $t\in\mathbb{N}$, then we can assert that this property also holds for the differential cost-to-go function. Therefore, our goal is to show that
\begin{equation}
V_{t}(S_2)\geq V_{t}(S_1), \quad \forall t\in\mathbb{N},\forall S_2\geq S_1>0.
\label{goalinduction}
\end{equation}
Note that we restrict our attention to non-zero states since $V_t(0)=0$ for any $t\in\mathbb{N}$. We prove the non-decreasing property reported in (\ref{goalinduction})  by induction. First, given that $V_0(S)=0$ for all states $S\in\mathbb{N}$, the above property holds for $t=0$. Next, we suppose that the property holds up till iteration $t>0$. By investigating eq. (\ref{bellmanequation}) for $S=0$, we can see that the optimal action is to stay idle. Therefore, we have
\begin{equation}
T(V_t)(0)=f(0)+(1-\alpha)V_{t}(1).
\end{equation}
Therefore, we can rewrite the update rule of the RVIA as
\begin{equation}
V_{t+1}(S)=T(V_t)(S)-f(0)-(1-\alpha)V_{t}(1), \quad \forall S\in\mathbb{N}^*.
\label{finalupdatervia}
\end{equation}
Next, given the system's dynamics reported in Section \ref{penaltydynamics}, we can conclude that
\begin{align}
T(V_t)(S_1)=&\min\big\{f(S_1)+\beta V_t(S_1+1); f(S_1)\nonumber\\&+\lambda+aV_t(S_1+1)\big\},\nonumber\\
T(V_t)(S_2)=&\min\big\{f(S_2)+\beta V_t(S_2+1);f(S_2)\nonumber\\&+\lambda+aV_t(S_2+1)\big\}.
\end{align}
Using the above equations, and by leveraging our assumption on $V_t(\cdot)$ and the non-decreasing property of $f(\cdot)$, we can deduce that $V_{t+1}(S_2)\geq V_{t+1}(S_1), \:\:\forall t\in\mathbb{N},\forall S_2\geq S_1>0$.

Concerning the bounded function case, we first note that the equations in (\ref{bellmanequation}) hold for any $S\in\mathbb{S}\setminus\{S_{\textnormal{thresh}}\}$. Moreover, we have 
\begin{align}
\theta + V(S_{\textnormal{thresh}})=&\min\big\{f(S_{\textnormal{thresh}})+\beta V(S_{\textnormal{thresh}})+(1-\beta)V(0);\nonumber\\&f(S_{\textnormal{thresh}})+\lambda+aV(S_{\textnormal{thresh}})+(1-a)V(0)\big\}.
\end{align}
By following the same analysis as the one done in the unbounded case, we can prove that $V(\cdot)$ is also non-decreasing in the bounded function case. This concludes our proof.
\section{Proof of Proposition \ref{thresholdproposition}}
\label{proofthresholdproposition}
Let us first focus on the unbounded function case. To establish the optimal policy of problem (\ref{secondobjective}), one has to recourse to solving the Bellman equation. However, without any knowledge of the optimal policy structure, deriving a closed-form expression of $V(\cdot)$ can be challenging. To address these challenges, we recall that the RVIA allows us to find the differential cost-to-go function iteratively. To that end, let us define $V^1_{t+1}(S)$ and $V^0_{t+1}(S)$ as the differential cost-to-go function estimate by the RVIA at iteration $t+1$ if the optimal action is $\psi=1$ and $\psi=0$ respectively. Given the RVIA update rule reported in (\ref{finalupdatervia}), we have 
\begin{equation}
V^1_{t+1}(S)=f(S)+\lambda+aV_t(S+1)-f(0)-(1-\alpha)V_{t}(1), \:\: \forall S\in\mathbb{N}^*,
\end{equation}
\begin{equation}
V^0_{t+1}(S)=f(S)+\beta V_t(S+1)-f(0)-(1-\alpha)V_{t}(1), \quad \forall S\in\mathbb{N}^*.
\end{equation}
Next, we let $\Delta V_{t+1}(S)=V^1_{t+1}(S)-V^0_{t+1}(S)$. Therefore, we have
\begin{equation}
\Delta V_{t+1}(S)=\lambda+(a-\beta)V_t(S+1), \quad \forall S\in\mathbb{N}^*.
\end{equation}
By definition, the sign of $\Delta V_{t+1}(S)$ allows us to conclude the optimal action that minimizes the Right Hand Side (\textbf{RHS}) of the update rule reported in (\ref{finalupdatervia}). For example, if $\Delta V_{t+1}(S)\geq 0$, then the minimum of the RHS in (\ref{finalupdatervia}) is achieved for $\psi=0$ and vice-versa. Note that, as we explained previously in Section \ref{penaltydynamics}, we have $a<\beta$. Moreover, we recall the 
results of Lemma \ref{nondecreasingproperty} where we have shown that $V_t(S+1)$ is a non-decreasing function of $S$ for all $t\in\mathbb{N}$. With these two things in mind, we can conclude that $\Delta V_{t+1}(S)$ is nothing but the sum of a non-negative constant $\lambda$, and a non-increasing negative function $(a-\beta)V_t(S+1)$. Knowing that the RVIA converges to the differential cost-to-go function $V(\cdot)$ when $t\rightarrow +\infty$, we can deduce that the optimal action is increasing with $S$ from $\psi=0$ to $\psi=1$. In other words, the difference $\Delta V(S)$ decreases with $S$, and at a certain point, it could change sign and becomes negative. When that happens, the action of transmitting becomes more beneficial than remaining idle. Therefore, we can conclude that the optimal transmission policy is of a threshold nature.

As for the bounded function case, the same analysis holds, and $\Delta V_{t+1}(S)$ is the sum of a non-negative constant $\lambda$ and a non-increasing negative function for any $S\in\mathbb{S}$. Accordingly, the difference $\Delta V(S)$ also decreases with $S$ and, at a certain point, it could change sign and become negative. When that happens, the action of transmitting becomes more beneficial than remaining idle. However, the subtle difference with the unbounded function case is that the sign's change might not happen. In this case, the optimal policy is to never transmit a packet. This is a natural consequence of the finite state space assumption resulting from the boundedness of the function. In fact, $\lambda$ can be significantly high that letting the system evolve on its own becomes optimal.
\section{Proof of Theorem \ref{maintheorem}}
\label{proofofmaintheorem}
\begin{table*}[ht]
\centering
\begin{tabular}{p{0.99\linewidth}}
\begin{equation*}
V(S)=\begin{cases}
0  & \text{if} \:\:  S=0 \\
\frac{-\theta'_n(1-\beta^{n-S})}{1-\beta}+\sum\limits_{k=1}^{n-S}f(n-k)\beta^{n-S-k}+\beta^{n-S}V(n) & \text{if} \:\: 1\leq S\leq n-1 \\
              \frac{-\theta'_n+\lambda}{1-a}+\sum\limits_{k=S}^{S_{\textnormal{thresh}}-1}f(S)a^{k-S}+\frac{a^{S_{\textnormal{thresh}}-S}f(S_{\textnormal{thresh}})}{1-a} & \text{if} \:\:  n\leq S\leq S_{\textnormal{thresh}}-1 \\
\frac{-\theta'_n+\lambda}{1-a}+ \frac{f(S_{\textnormal{thresh}})}{1-a} & \text{if} \:\:  S=S_{\textnormal{thresh}}
       \end{cases} 
\end{equation*}
\begin{equation*}
\theta'_n=\frac{\frac{f(0)}{1-\alpha}+\sum\limits_{k=1}^{n-1}f(n-k)\beta^{n-k-1}+\beta^{n-1}\sum\limits_{k=n}^{S_{\textnormal{thresh}}-1}f(k)a^{k-n}+\frac{\lambda \beta^{n-1}}{1-a}+\beta^{n-1}\frac{a^{S_{\textnormal{thresh}}-S}f(S_{\textnormal{thresh}})}{1-a}}{\frac{1}{1-\alpha}+\frac{1-\beta^{n-1}}{1-\beta}+\frac{\beta^{n-1}}{1-a}}
\end{equation*}
\\
\hline
\vspace{-5pt}
\caption{Expressions of $V(S)$ and $\theta'_n$ for the bounded function case.}
\label{allequations}
\end{tabular}
\vspace{-25pt}
\end{table*}
As always, we start by investigating the unbounded function case. Given that the optimal policy is a threshold policy, we can affirm that an integer value $n\in\mathbb{N}$ exists such that the optimal action is $\psi=1$ and $\psi=0$ when $S\geq n$ and $S<n$ respectively. With that in mind, and by utilizing the RHS of the Bellman equation in (\ref{bellmanequation}), we can conclude that
\begin{align}
f(S)+\beta V(S+1)+&(1-\beta)V(0)>f(S)+\lambda\nonumber\\&+aV(S+1)+(1-a)V(0), \quad \forall S\geq n.
\end{align}
Without loss of generality, we suppose in the sequel that $V(0)=0$. To that end, and by rearranging the above terms, the following condition for activity can be deduced
\begin{equation}
V(S+1)>\frac{\lambda}{\beta-a}.
\label{conditionactivity}
\end{equation} 
In other words, the optimal action is to transmit whenever the system is in a state $S$ that verifies the above condition. Given the threshold property of the optimal policy, the Bellman equation can be rewritten for any state $S\geq n$ as follows
\begin{equation}
V(S)=-\theta_n+\lambda+f(S)+aV(S+1), \quad \forall S\geq n.
\end{equation}
Note that we add the subscript $n$ to $\theta$ to indicate that the average cost $\theta$ results from adopting the threshold $n$. By following a forward induction, we obtain
\begin{equation}
V(S)=(-\theta_n+\lambda)(1+a+a^2+\ldots)+\sum\limits_{k=S}^{+\infty}f(k)a^{k-S}, \quad \forall S\geq n.
\label{largerthannvalue}
\end{equation}
Given that $a<1$, we can invoke the geometric series sum property to end up with
\begin{equation}
V(S)=\frac{-\theta_n+\lambda}{1-a}+\sum\limits_{k=S}^{+\infty}f(k)a^{k-S}, \quad \forall S\geq n.
\label{largerthannvalue}
\end{equation}
Given the above equation, we can particularly conclude that 
\begin{equation}
V(n)=\frac{-\theta_n+\lambda}{1-a}+\sum\limits_{k=n}^{+\infty}f(k)a^{k-n}.
\label{Vnexpression}
\end{equation}
Next, we investigate the case where the system is in a state $S<n$. For any state $S<n$, the optimal action is to remain idle. Hence, using the the Bellman equation, we obtain
\begin{equation}
V(S)=\begin{cases}
              -\theta_n+f(S)+\beta V(S+1), & \text{if} \:\:  1\leq S<n, \\
             
            \frac{-\theta_n+f(0)}{1-\alpha}+V(1), & \text{if} \:\: S=0.
       \end{cases} 
\label{smallerthannvalue}
\end{equation}
By following a backward induction, we wind up with the following identity for any $1\leq S<n$
\begin{equation}
V(S)=\frac{-\theta_n(1-\beta^{n-S})}{1-\beta}+\sum\limits_{k=1}^{n-S}f(n-k)\beta^{n-S-k}+\beta^{n-S}V(n).
\label{backwardinduction}
\end{equation}
Knowing that $V(0)=0$, and by using eq. (\ref{smallerthannvalue}), we get 
\begin{equation}
V(1)=\frac{\theta_n-f(0)}{1-\alpha}.
\label{V1expression}
\end{equation}
Using the expression of $V(n)$ in (\ref{Vnexpression}), and by replacing $S$ with $1$ in eq. (\ref{backwardinduction}) and equating it to $V(1)$ in (\ref{V1expression}), we end up with the following relationship between $\theta_n,\lambda$, and $n$
\begin{equation}
\theta_n=\frac{\frac{f(0)}{1-\alpha}+\sum\limits_{k=1}^{n-1}f(n-k)\beta^{n-k-1}+\beta^{n-1}\sum\limits_{k=n}^{+\infty}f(k)a^{k-n}+\frac{\lambda \beta^{n-1}}{1-a}}{\frac{1}{1-\alpha}+\frac{1-\beta^{n-1}}{1-\beta}+\frac{\beta^{n-1}}{1-a}}
\label{thetaexpression}
\end{equation}
This fundamental relationship will be pivotal to our subsequent analysis to find the threshold $n$. The next step revolves around deriving a criterion that will allow us find $n$. To that end, we recall the activity condition reported in (\ref{conditionactivity}). Given that $n$ is the threshold, we can assert that
\begin{equation}
V(n)\leq \frac{\lambda}{\beta-a}<V(n+1).
\label{conditionnn}
\end{equation}
Therefore, it is sufficient to find the value $n$ that verifies the above equation. This is however easier said than done as one has to prove the existence of such a solution. To proceed in this direction, we recall the results of Lemma \ref{nondecreasingproperty} where we have shown that $V(\cdot)$ is a non-decreasing function. With that in mind, we recall that the function $f(\cdot)$ is unbounded. Therefore, by leveraging the limit definition, we have 
\begin{equation}
\forall M>0,\:\exists S_0: \forall S\geq S_0,  \: f(S)>M.
\end{equation}
Using the above property of $f(\cdot)$ and the expression of $V(n)$ in (\ref{Vnexpression}), we can show that 
\begin{equation}
\forall M'>0,\:\exists S'_0: \forall S'\geq S'_0,  \: V(S')>M'.
\end{equation}
Therefore, a solution to eq. (\ref{conditionnn}) exists in this case. In particular, the optimal threshold is 
\begin{align}
n^*&=\sup\{n\in\mathbb{N}:V(n)\leq \frac{\lambda}{\beta-a}\}=\sup\{n\in\mathbb{N}:\nonumber\\
&\frac{-\theta_n(\beta-a)+\lambda(\beta-1)}{(1-a)(\beta-a)}+\sum\limits_{k=n}^{+\infty}f(k)a^{k-n}\leq0\}\nonumber\\
&=\inf\{n\in\mathbb{N}^*:H(n)>0\}-1,
\label{optimalthresholdexpression}
\end{align}
\begin{equation}
H(n)=\frac{-\theta_n(\beta-a)+\lambda(\beta-1)}{(1-a)(\beta-a)}+\sum\limits_{k=n}^{+\infty}f(k)a^{k-n}.
\end{equation}
To understand the intuition behind these results, we recall that $\lambda$ can be seen as a penalty paid for transmitting a packet. As $f(\cdot)$ is unbounded when $S\rightarrow+\infty$, we can deduce that no matter how high $\lambda$ is, transmitting a packet will \textbf{eventually} become the optimal action. 

Let us now investigate the bounded function case. To that end, similarly to the previous case, we suppose that the optimal threshold is equal to $n\in\mathbb{S}$. Following the same analysis above, we end up with the expressions of $V(.)$ and $\theta'_n$ reported in Table \ref{allequations}. Moreover, from the Bellman equation, we can conclude that the activity condition is
\begin{align}
&V(S+1)>\frac{\lambda}{\beta-a}, \quad \forall S\in\mathbb{S}\setminus\{S_{\textnormal{thresh}}\},\nonumber\\&V(S_{\textnormal{thresh}})>\frac{\lambda}{\beta-a}, \quad S=S_{\textnormal{thresh}}.
\label{conditionactivitybounded}
\end{align} 
Given the above condition, we can conclude that if it is optimal to transmit when $S=S_{\textnormal{thresh}}$, then it is also optimal to transmit when $S=S_{\textnormal{thresh}}-1$. Accordingly, we focus on $n$ being in the set $\mathbb{S}\setminus \{S_{\textnormal{thresh}}\}$. With the above activity condition in mind, the threshold $n\in\mathbb{S}\setminus \{S_{\textnormal{thresh}}\}$ is simply the first state that verifies $V(n+1)>\frac{\lambda}{\beta-a}$. In other words,
\begin{equation}
n^*=\inf\{n\in\mathbb{S}\setminus\{0\}:H'(n)>0\}-1,
\label{optimalthresholdbounded}
\end{equation}
\begin{align}
H'(n)=&\frac{-\theta'_n(\beta-a)+\lambda(\beta-1)}{(1-a)(\beta-a)}+\sum\limits_{k=n}^{S_{\textnormal{thresh}}-1}f(k)a^{k-n}+\nonumber\\&\frac{a^{S_{\textnormal{thresh}}-S}f(S_{\textnormal{thresh}})}{1-a}.
\end{align}
Now, unlike the unbounded function case, an interesting phenomenon can take place here: the activity penalty $\lambda$ can be so high that it is optimal to simply not transmit, even if $S$ is high. In other words, transmitting a packet will cost us more than letting the system evolve on its own without any intervention. Our aim becomes to characterize this regime and derive a condition on $\lambda$ to know when this phenomenon occurs. If a threshold exists, it can be found using eq. (\ref{optimalthresholdbounded}). Therefore, if $H'(S_{\textnormal{thresh}})\leq0$, then the optimal policy is to stay idle. In other words, if
\begin{equation}
\lambda\geq \frac{(\beta-a)(f(S_{\textnormal{thresh}})-\theta'_{S_{\textnormal{thresh}}})}{1-\beta},
\end{equation}
then the optimal policy is to stay idle. On the other hand, if $\lambda$ does not verify the above inequality, then the optimal policy is a threshold policy where the threshold can be found using eq. (\ref{optimalthresholdbounded}).

\section{Proof of Proposition \ref{propositionupdaterate}}
\label{proofpropositionupdaterate}
To proceed with our proof, we recall that the optimal transmission policy $\pi^*_\lambda$ is a threshold policy with a threshold $n^*_{\lambda}$. Trivially, if $n^*_{\lambda}=0$, a packet transmission is initiated at each timeslot and $C_{\pi^*_{\lambda}}=1$. In the case where $n^*_{\lambda}>0$, we note that the system's state $S_t$ evolves as a Discrete-Time Markov Chain (\textbf{DTMC}) reported in Fig. \ref{thresholdchain}. Note that we first focus on the case of unbounded function $f(\cdot)$. By leveraging the general balance equations, we can show that the stationary distribution of the DTMC is
\begin{align}
\sigma_0(n^*_{\lambda})&=\frac{1}{1+\frac{(1-\alpha)(1-\beta^{n^*_{\lambda}})}{1-\beta}+\frac{(1-\alpha)a\beta^{n^*_{\lambda}-1}}{1-a}},\nonumber\\
\sigma_k(n^*_{\lambda})&=\begin{cases}
(1-\alpha)\beta^{k-1}\sigma_0(n^*_{\lambda}), \: &\textnormal{if} \:\: 1\leq k\leq n^*_{\lambda},\\
(1-\alpha)\beta^{n^*_{\lambda}-1}a^{k-n^*_{\lambda}}\sigma_0(n^*_{\lambda}), \: &\textnormal{if} \:\: k\geq n^*_{\lambda}+1.\\
\end{cases}
\end{align}
Given the above expressions, and knowing that $C_{\pi^*_{\lambda}}=\sum\limits_{k=n}^{+\infty}\sigma_k(n^*_{\lambda})$, we can obtain the results of the proposition. By following a similar analysis for the bounded function case, we can show that the stationary distribution has the following expression
\begin{align}
\sigma_0(n^*_{\lambda})&=\frac{1}{1+\frac{(1-\alpha)(1-\beta^{n^*_{\lambda}})}{1-\beta}+\frac{(1-\alpha)a\beta^{n^*_{\lambda}-1}}{1-a}},\nonumber\\
\sigma_k(n^*_{\lambda})&=\begin{cases}
(1-\alpha)\beta^{k-1}\sigma_0(n^*_{\lambda}),  &\hspace{-10pt}\textnormal{if} \:\: 1\leq k\leq n^*_{\lambda},\\
(1-\alpha)\beta^{n^*_{\lambda}-1}a^{k-n^*_{\lambda}}\sigma_0(n^*_{\lambda}),  &\hspace{-22pt}\textnormal{if} \:\: n^*_{\lambda}+1 \leq k\leq  S_{\textnormal{thresh}}\\
(1-\alpha)\beta^{n^*_{\lambda}-1}\frac{a^{S_{\textnormal{thresh}}-n^*_{\lambda}}}{1-a}
\sigma_0(n^*_{\lambda}),  &\hspace{-10pt}\textnormal{if} \:\:  k=S_{\textnormal{thresh}}.
\end{cases}
\end{align}
Next, we demonstrate that the average update rate $C_{\pi^*_{\lambda}}=\sum\limits_{k=n}^{S_{\textnormal{thresh}}}\sigma_k(n^*_{\lambda})$ has the same expression as the unbounded case for $n^*_{\lambda}\in\mathbb{S}$. Note that, the average update rate is equal to $0$ when $n^*_{\lambda}>S_{\textnormal{thresh}}$.
\begin{figure}[t]
\centering
\includegraphics[width=.99\linewidth]{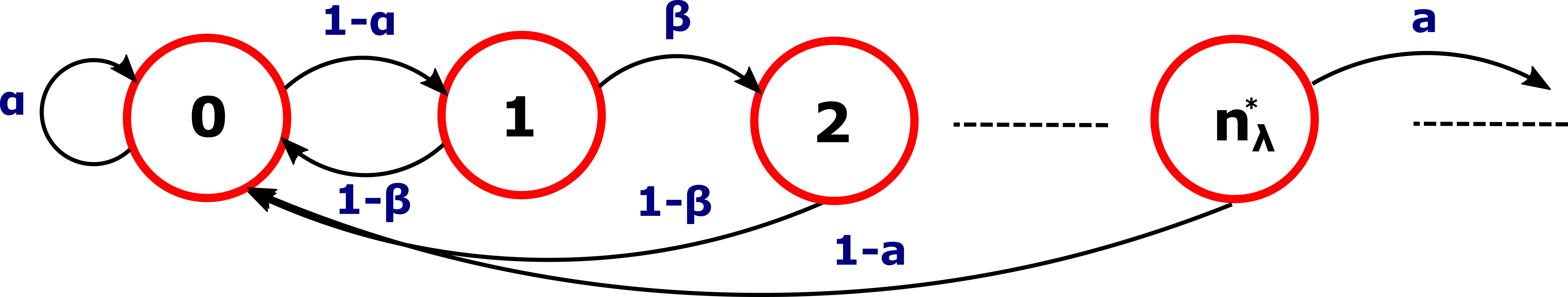}
\caption{The states transitions given a fixed threshold.}
\label{thresholdchain}
\setlength{\belowcaptionskip}{-15pt}
\end{figure}
\vspace{-10pt}
\section{Proof of Theorem \ref{theoremrelatingtoconstrained}}
\label{prooftheoremrelatingtoconstrained}
Let us first study the unbounded function case. To establish our theorem, we need to show that the constrained problem reported in (\ref{originalobjective}) verifies key properties listed in the assumptions of \cite[Theorem~2.5]{sennott1993}. To do so, let $\mathfrak{R}(s,G)$ be the class of policies such that
\begin{equation}
\Pr(S_t\in G \textnormal{ for some $t\geq0$}|S_0=s)=1,
\end{equation}
and the expected time $m_{sG}$ of a first passage from $s$ to $G$ is finite. Let $\mathfrak{R}^*(s,G)$ be the class of policies $\pi\in \mathfrak{R}(s,G)$ such that, in addition, the expected AoII and update cost of a first passage from $s$ to G are finite. We can now prove that our problem verifies the assumptions.

\textbf{Assumption 1} - For all $r>0$, the set $G(r)=\{s:$ there exists an action $\psi$ such that $f(s) + \psi\leq r\}$ is finite: To prove this assumption, we note that $f(\cdot)$ is a non-decreasing unbounded function. Accordingly, given that $\lim_{S \to +\infty} f(S)=+\infty$, we have
\begin{equation}
\forall M>0,\:\exists S': \forall S\geq S',  \: f(S)>M.
\end{equation}
Given that $\psi\in\{0,1\}$, we set $M=r$ to deduce that $G(r)\subseteq[0,S'-1]$, which is a finite set.

\textbf{Assumption 2} -  There exists a stationary policy $\pi$ such that it induces a Markov chain where the state space consists of a single (nonempty) positive recurrent class $R$ and a set $U$ of transient states such that $\pi\in\mathfrak{R}^*(i,R)$, for $i\in U$, and both the average AoII and update rate are finite:  To prove this assumption, we consider the always update policy $\pi_{\textnormal{au}}$ that transmits a packet at each timeslot. Given the system's dynamics reported in Section \ref{penaltydynamics}, we can conclude that the state space of the Markov chain induced by this policy consists of a single recurrent class $R=\mathbb{N}$ (the transient set $U$ is empty). Moreover, we have $C_{\pi}=1$ and, given the assumption on $f(\cdot)$ found in (\ref{conditiononf}), we can deduce that the average AoII of $\pi_{\textnormal{au}}$ is finite. 

\textbf{Assumption 3} -  Given any two states $S_1\neq S_2$, there exists a policy $\pi$ such that $\pi\in\mathfrak{R}^*(i,j)$:  To prove this assumption, let us suppose without loss of generality that $S_2\geq S_1$. By considering the always update policy, we can see that there is a non-zero probability to go from state $S_1$ to state $S_2$ and vice-versa. The expected AoII and update costs of the first passage from $S_1$ to $S_2$ (or vice-versa) are trivially finite.     

\textbf{Assumption 4} -  If a stationary policy $\pi$ has at least one positive recurrent state,
then it has a single positive recurrent class $R$. Moreover, if $0\not\in R$, then $\pi\in\mathfrak{R}^*(0,R)$: To show this, we simply note that whatever the transmission policy is, there is a non-zero probability to go from any state $S\in\mathbb{N}^*$ to state $0$ and vice-versa. Therefore, any recurrent class must contain the state $0$. Hence, we can conclude that there can only be one single positive recurrent class. 

\textbf{Assumption 5} - There exists a policy $\pi$ such that the average AoII is finite and $C_{\pi}<\delta$: To show this, we can consider a threshold policy $\pi_{n_0}$ where the threshold $n_0=\inf\{n:\mathbb{N}:C_{\pi_n}<\delta\}$. Note that the update rate $C_{\pi_n}$ is strictly decreasing with $n$ \cite{sennott1993}, which ensures the existence of $n_0$. Given the assumption on $f$ found in (\ref{conditiononf}), we can conclude that the AoII is finite. 

Given the above assumptions, we can leverage the results of \cite{sennott1993} (in particular, Theorem 2.5, Proposition 3.2, Lemma 3.4, and Lemma 3.9). These results affirm that the optimal transmission policy of the constrained problem is a mixture of two policies such that\begin{itemize}
\item The two policies coincide with those of the optimal policy of problem (\ref{firstobjective}) for a certain $\lambda^*\geq 0$, but differ in at most a single state.
\item $\lambda^*$ is defined as $\lambda^*\triangleq\inf\{\lambda\in\mathbb{R}^+: C_{\pi^*_{\lambda^*}}\leq\delta\}$.
\item The parameter $\mu^*\in[0,1]$ ensures that the update rate constraint is verified with equality.
\end{itemize}
Given the above results, we can conclude the statements of our theorem.

As for the bounded function case, we first discuss the validity of Hypothesis 2.2 and Hypothesis 4.1 of \cite{BEUTLER1985236} for our problem. To that end, we have:

\textbf{Hypothesis 2.2.} - For any stationary policy, the state $0$ is accessible from any $S\in\mathbb{S}$: This hypothesis holds for our problem as seen in the system's dynamics reported in Section \ref{penaltydynamics}. 

\textbf{Hypothesis 4.1.} - Let $\tilde{\Pi}$ denote the set of optimal policies for the unconstrained version of the problem in eq. (\ref{originalobjective}). Suppose that $C_{\tilde{\pi}}>\delta$ for every $\tilde{\pi}\in\tilde{\Pi}$ and that there exists a stationary policy $\hat{\pi}$ such that $C_{\hat{\pi}}<\delta$: First, it is easy to see that the never transmit policy $\hat{\pi}$ has an average update rate $C_{\hat{\pi}}=0$. Next, a careful investigation of this hypothesis is needed as there could be cases where $C_{\tilde{\pi}}\leq\delta$. To see this more clearly, consider the stationary policy $\tilde{\pi}$ where a transmission is initiated only when $S_t\neq0$. By using the expression provided in Proposition \ref{propositionupdaterate}, it can be shown that $C_{\tilde{\pi}}=\frac{1-\alpha}{2-\alpha-a}$. Moreover, given the system's dynamics, the Bellman equation in state $0$ for the unconstrained version of the problem in eq. (\ref{originalobjective}) can be written as follows
\begin{align}
\theta' + V'(0)=&\min\big\{f(0)+\alpha V'(0)+(1-\alpha)V'(1);f(0)\nonumber\\&+\alpha V'(0)+(1-\alpha)V'(1)\big\}.
\end{align}
In other words, transmitting a packet in state $S=0$ does not have any impact on the performance. Therefore, if $\delta\geq \frac{1-\alpha}{2-\alpha-a}$, then the constraint becomes redundant and the AoII optimal policy can be obtained by transmitting whenever $S\neq0$. \\
Now, let us focus on the case where $\delta<\frac{1-\alpha}{2-\alpha-a}$. In this case, the two hypotheses hold. Let us define $\lambda^*\triangleq\inf\{\lambda\in\mathbb{R}^+: C_{\pi^*_{\lambda^*}}\leq\delta\}$. By using Theorem 4.4 \cite{BEUTLER1985236}, we can deduce that the optimal transmission policy of the constrained problem is a mixture of two policies such that
\begin{itemize}
\item The two policies coincide with those of the optimal policy of problem (\ref{firstobjective}) for $\lambda^*\geq 0$, but differ in at most a single state.
\item The parameter $\mu^*\in[0,1]$ ensures that the update rate constraint is verified with equality.
\end{itemize}
Given the above results, we can conclude the statements of the theorem.

\newpage
\onecolumn
\section{Algorithms pseudo-code}
\label{algorithmpseudocode}
\begin{algorithm}
\caption{AoII Optimal Policy - Unbounded Function}\label{euclid}
\begin{algorithmic}[1]
\State \textbf{Input}: the system's parameters $\alpha, \beta, p_s, \delta$ and the convergence tolerance $\epsilon$
\If {$\delta=1$} skip the algorithm and transmit at every timeslot $t$
\Else 
\State \textbf{Init.} $\lambda_{\textnormal{min}}\leftarrow0$, $\lambda_{\textnormal{max}}\leftarrow1$
\State $n^*_{\lambda_{\textnormal{max}}}\leftarrow$ FindThreshold$(\alpha, \beta, p_s, \lambda_{\textnormal{max}})$
\State  $C\leftarrow C_{\pi^*_{\lambda_{\textnormal{max}}}}$  using Proposition \ref{propositionupdaterate}
\While {$C>\delta$}
\State $\lambda_{\textnormal{min}}\leftarrow\lambda_{\textnormal{max}}$, $\lambda_{\textnormal{max}}\leftarrow 2\lambda_{\textnormal{max}}$
\State $n^*_{\lambda_{\textnormal{max}}}\leftarrow$ FindThreshold$(\alpha, \beta, p_s, \lambda_{\textnormal{max}})$
\State  $C\leftarrow C_{\pi^*_{\lambda_{\textnormal{max}}}}$  using Proposition \ref{propositionupdaterate}
\EndWhile
\State $\xi\leftarrow\frac{\lambda_{\textnormal{min}}+\lambda_{\textnormal{max}}}{2}$
\While {$|\xi-\lambda_{\textnormal{max}}|>\epsilon$}
\State $n^*_{\xi}\leftarrow$ FindThreshold$(\alpha, \beta, p_s, \xi)$
\State $C\leftarrow C_{\pi^*_{\xi}}$  using Proposition \ref{propositionupdaterate}
\If {$C>\delta$} $\lambda_{\textnormal{min}}\leftarrow\xi$ 
\Else $\:\:\lambda_{\textnormal{max}}\leftarrow\xi$
\EndIf
\EndWhile
\State $\lambda^*\leftarrow\xi$
\If {$C>\delta$} $n^*_{\lambda^*}\leftarrow n^*+1$, $C_{\pi^*_{\lambda^*,1}}\leftarrow C$, $C_{\pi^*_{\lambda^*,2}}\leftarrow\frac{(1-\alpha)\beta^{n^*_{\lambda^*}-1}}{(1-a)(1+\frac{(1-\alpha)(1-\beta^{n^*_{\lambda^*}})}{1-\beta}+\frac{(1-\alpha)a\beta^{n^*_{\lambda^*}-1}}{1-a})}$
\Else $\:\:n^*_{\lambda^*}\leftarrow n^*$, $C_{\pi^*_{\lambda^*,2}}\leftarrow C$
   \If {$n^*=1$} $C_{\pi^*_{\lambda^*,1}}\leftarrow 1$
\Else $\:\: n^*\leftarrow n^*-1$
\State $\:\:C_{\pi^*_{\lambda^*,1}}\leftarrow\frac{(1-\alpha)\beta^{n^*-1}}{(1-a)(1+\frac{(1-\alpha)(1-\beta^{n^*})}{1-\beta}+\frac{(1-\alpha)a\beta^{n^*-1}}{1-a})}$
\EndIf
\EndIf
\State $\mu^*\leftarrow\frac{\delta-C_{\pi^*_{\lambda^*,2}}}{C_{\pi^*_{\lambda^*,1}}-C_{\pi^*_{\lambda^*,2}}}$
\State \textbf{Output}: $n^*_{\lambda^*},\mu^*$
\EndIf
\Procedure{FindThreshold}{$\alpha, \beta, p_s, \lambda$}
\State \textbf{Init.} $N_{LB}\leftarrow 1, N_{UB}\leftarrow 1$
\While {$H(N_{UB})\leq0$}
\State $N_{LB}\leftarrow N_{UB}, N_{UB}\leftarrow 2N_{UB}$
\EndWhile
\State $n'\leftarrow\ceil*{\frac{N_{LB}+N_{UB}}{2}} $
\While {$n'<N_{UB}$}
\If {$H(n')\leq0$} $N_{LB}\leftarrow n'$
\Else $\:\:N_{UB}\leftarrow n'$
\EndIf
\State $n'\leftarrow\ceil*{\frac{N_{LB}+N_{UB}}{2}} $
\EndWhile
\State \textbf{Output}: the optimal threshold $n^*_{\lambda}\leftarrow n'-1$
\EndProcedure
\end{algorithmic}
\end{algorithm}
\setlength{\textfloatsep}{0pt}
\onecolumn
\begin{algorithm}
\caption{AoII Optimal Policy - Bounded Function}\label{euclid}
\begin{algorithmic}[1]
\State \textbf{Input}: the system's parameters $\alpha, \beta, p_s, \delta, S_{\textnormal{thresh}}$ and the convergence tolerance $\epsilon$
\If {$\delta\geq \frac{1-\alpha}{2-\alpha-a}$} skip the algorithm and transmit at every timeslot $t$ when $S_t\neq0$
\Else 
\State \textbf{Init.} $\lambda_{\textnormal{min}}\leftarrow0$, $\lambda_{\textnormal{max}}\leftarrow1$
\State $n^*_{\lambda_{\textnormal{max}}}\leftarrow$ FindThreshold$(\alpha, \beta, p_s, \lambda_{\textnormal{max}})$
\State  $C\leftarrow C_{\pi^*_{\lambda_{\textnormal{max}}}}$  using Proposition \ref{propositionupdaterate}
\While {$C>\delta$}
\State $\lambda_{\textnormal{min}}\leftarrow\lambda_{\textnormal{max}}$, $\lambda_{\textnormal{max}}\leftarrow 2\lambda_{\textnormal{max}}$
\State $n^*_{\lambda_{\textnormal{max}}}\leftarrow$ FindThreshold$(\alpha, \beta, p_s, \lambda_{\textnormal{max}}), \:\: C\leftarrow C_{\pi^*_{\lambda_{\textnormal{max}}}}$ using Proposition \ref{propositionupdaterate}

\EndWhile
\State $\xi\leftarrow\frac{\lambda_{\textnormal{min}}+\lambda_{\textnormal{max}}}{2}$
\While {$|\xi-\lambda_{\textnormal{max}}|>\epsilon$}
\State $n^*_{\xi}\leftarrow$ FindThreshold$(\alpha, \beta, p_s, \xi), \:\: C\leftarrow C_{\pi^*_{\xi}}$ using Proposition \ref{propositionupdaterate}
\If {$C>\delta$} $\lambda_{\textnormal{min}}\leftarrow\xi$ 
\Else $\:\:\lambda_{\textnormal{max}}\leftarrow\xi$
\EndIf
\EndWhile
\State $\lambda^*\leftarrow\xi$
\If {$C>\delta$} $n^*_{\lambda^*}\leftarrow n^*+1$, $C_{\pi^*_{\lambda^*,1}}\leftarrow C$, 
\If {$n^*_{\lambda^*}=S_{\textnormal{thresh}}+1$} $C_{\pi^*_{\lambda^*,2}}\leftarrow0$
\Else $\:\:C_{\pi^*_{\lambda^*,2}}\leftarrow\frac{(1-\alpha)\beta^{n^*_{\lambda^*}-1}}{(1-a)(1+\frac{(1-\alpha)(1-\beta^{n^*_{\lambda^*}})}{1-\beta}+\frac{(1-\alpha)a\beta^{n^*_{\lambda^*}-1}}{1-a})}$
\EndIf
\Else $\:\:n^*_{\lambda^*}\leftarrow n^*$, $C_{\pi^*_{\lambda^*,2}}\leftarrow C$
   \If {$n^*=1$} $C_{\pi^*_{\lambda^*,1}}\leftarrow 1$
\Else $\:\: n^*\leftarrow n^*-1$
\State $\:\:C_{\pi^*_{\lambda^*,1}}\leftarrow\frac{(1-\alpha)\beta^{n^*-1}}{(1-a)(1+\frac{(1-\alpha)(1-\beta^{n^*})}{1-\beta}+\frac{(1-\alpha)a\beta^{n^*-1}}{1-a})}$
\EndIf
\EndIf
\State $\mu^*\leftarrow\frac{\delta-C_{\pi^*_{\lambda^*,2}}}{C_{\pi^*_{\lambda^*,1}}-C_{\pi^*_{\lambda^*,2}}}$
\State \textbf{Output}: $n^*_{\lambda^*},\mu^*$
\EndIf
\Procedure{FindThreshold}{$\alpha, \beta, p_s, \lambda$}
\If {$\lambda\geq\frac{(\beta-a)(f(S_{\textnormal{thresh}})-\theta_{S_{\textnormal{thresh}}})}{1-\beta}$} $n^*_{\lambda}\leftarrow S_{\textnormal{thresh}}+1$
\Else 
\State \textbf{Init.} $N_{LB}\leftarrow 1, N_{UB}\leftarrow 1$
\While {$H'(N_{UB})\leq0$}
\State $N_{LB}\leftarrow N_{UB}, N_{UB}\leftarrow \min\{2N_{UB},S_{\textnormal{thresh}}\}$
\EndWhile
\State $n'\leftarrow\ceil*{\frac{N_{LB}+N_{UB}}{2}} $
\While {$n'<N_{UB}$}
\If {$H'(n')\leq0$} $N_{LB}\leftarrow n'$
\Else $\:\:N_{UB}\leftarrow n'$
\EndIf
\State $n'\leftarrow \ceil*{\frac{N_{LB}+N_{UB}}{2}} $
\EndWhile
\EndIf
\State \textbf{Output}: the optimal threshold $n^*_{\lambda}\leftarrow n'-1$
\EndProcedure
\end{algorithmic}
\end{algorithm}

\end{document}